\documentclass[onecolumn,draftcls,12pt]{IEEEtran}
\usepackage{amsmath,amsfonts}
\usepackage{algorithmic}
\usepackage{array}
\usepackage{textcomp}
\usepackage{stfloats}
\usepackage{url}
\usepackage{verbatim}
\usepackage{graphicx}
\def\BibTeX{{\rm B\kern-.05em{\sc i\kern-.025em b}\kern-.08em
    T\kern-.1667em\lower.7ex\hbox{E}\kern-.125emX}}
\usepackage{balance}
\usepackage{verbatim}
\usepackage{amsfonts}
\usepackage{amssymb}
\usepackage{stfloats}
\usepackage{cite}
\usepackage{soul}
\usepackage{setspace}
\usepackage{graphicx}
\usepackage{psfrag}
\usepackage{amsmath}
\usepackage{array}
\usepackage{epstopdf}
\usepackage{authblk}
\usepackage{graphicx} 
\usepackage{amsthm} 
\usepackage{lipsum}
\usepackage{verbatim} 
\usepackage{authblk}
\usepackage{mathtools}
\usepackage{cuted}
\usepackage[lined,boxed,ruled]{algorithm2e}
\usepackage{booktabs}
\usepackage{subfigure}

\usepackage[usenames]{color}

\newtheorem{Definition}{Definition}

\newtheorem{Lemma}{Lemma}
\newtheorem{Theorem}{Theorem}

\newtheorem{Remark}{Remark}

\newtheorem{Corollary}{Corollary}

\begin{document}

\title{Analog MIMO Communication for One-shot Distributed Principal Component Analysis}

\author{{Xu~Chen},~{Erik G. Larsson}, and~{Kaibin~Huang}

\thanks{X. Chen and K. Huang are with the Department of Electrical and Electronic Engineering, The University of Hong Kong, Hong Kong (Email: \{chenxu, huangkb\}@eee.hku.hk). E. G. Larsson is with the Department of Electrical Engineering (ISY), Linköping University, 58183 Linköping, Sweden (Email: \{erik.g.larsson\}@liu.se). Corresponding author: K. Huang.}}

\markboth{Transactions on Signal Processing,~Vol.~xx, No.~xx, JUNE~2022}%
{Analog MIMO Communication for One-shot Distributed Principal Component Analysis}

\maketitle
\begin{abstract}
A fundamental algorithm for data analytics at the edge of wireless networks is \emph{distributed principal component analysis} (DPCA), which finds the most important information embedded in a distributed high-dimensional dataset by distributed computation of a reduced-dimension data subspace, called \emph{principal components} (PCs). In this paper, to support one-shot DPCA in wireless systems, we propose a framework of analog MIMO transmission featuring the uncoded analog transmission of local PCs for estimating the global PCs. To cope with channel distortion and noise, two maximum-likelihood (global) PC estimators are presented corresponding to the cases with and without receive \emph{channel state information} (CSI). The first design, termed coherent PC estimator, is derived by solving a Procrustes problem and reveals the form of regularized channel inversion where the regulation attempts to alleviate the effects of both receiver noise and data noise. The second one, termed blind PC estimator, is designed based on the subspace channel-rotation-invariance property and computes a centroid of received local PCs on a Grassmann manifold. Using the manifold-perturbation theory, tight bounds on the \emph{mean square subspace distance} (MSSD) of both estimators are derived for performance evaluation. The results reveal simple scaling laws of MSSD concerning device population, data and channel \emph{signal-to-noise ratios} (SNRs), and array sizes. More importantly, both estimators are found to have identical scaling laws, suggesting the dispensability of CSI to accelerate DPCA. Simulation results validate the derived results and demonstrate the promising latency performance of the proposed analog MIMO.
\end{abstract}


\section{Introduction}

Mobile devices have  become  a dominant platform for Internet access and mobile data traffic has been growing at an exponential rate in the last decade. This motivates the relocation  of data analytics and machine learning algorithms originally performed in the cloud to the network edge to gain fast access to enormous  mobile data~\cite{NiyatoFL2020,GZhuMag2020}.  The distilled knowledge and trained AI models can  support a wide-range of mobile applications~\cite{BennisApplications2021}.  Among many data-driven techniques, \emph{principal component analysis} (PCA) is  a fundamental tool in data analytics that finds application in diverse scientific fields ranging from wireless communication (see e.g.,~\cite{WangPCA2020,SunPCA2020}) to machine learning (see e.g.,~\cite{BartlettPCA2002,BelhumeurPCA1997}). This unsupervised learning technique provides  a  simple way to identify a low-dimensional subspace, called \emph{principal components} (PCs), that contains the most important information of a high-dimensional dataset,  thereby facilitating  feature extraction and data compression~\cite{PCA}. Specifically, the principal components are computed  by~\emph{singular value decomposition} (SVD) of the data matrix comprising  data samples as its columns~\cite{gemp2021eigengame}. In a mobile network with distributed data, data uploading from edge devices (e.g., sensors and smartphones) for centralized PCA may not be feasible due to the issues of data privacy and ownership and uplink traffic congestion~\cite{FLsurvey2020}. This issue has motivated researchers to design   \emph{distributed PCA} (DPCA) algorithms~\cite{NIPS2014_DPCA}. A typical algorithm, called \emph{one-shot DPCA} and also considered in this work,  is to compute local PCs at each device  using its local data and then upload local solutions from multiple devices to a server for aggregation to give the \emph{global PCs}, which approximate the \emph{ground-truth}  solution corresponding to centralized PCA \cite{JF2019estimation2019, VC2020Arxiv, DG2017PMLR}. Alternatively, iterative DPCA can be designed by distributed implementation of classic numerical iterative algorithms, for example,   \emph{approximate Newton’s method}~\cite{chen2021distributed} and  \emph{stochastic gradient descent} \cite{zhang2021turning,MP2013Erratum}, which improves the performance of one-shot DPCA at the cost of much higher communication overhead. Fast  DPCA targets latency-sensitive applications such as  autonomous  driving, virtual reality,  and digital twins. To accelerate one-shot DPCA in a wireless network, we propose a novel framework of analog \emph{multiple-input-multiple-output} (MIMO) transmission. Its key feature is that each device directly transmits local PCs over a MIMO channel using uncoded  linear analog modulation, which requires no transmit \emph{channel state information} (CSI). We design optimal PC (subspace) estimators at the server for  both the cases with and without receive CSI. 

To support one-shot DPCA in a wireless system, the  proposed analog MIMO scheme involves the transmission of local PCs from each device as an uncoded and linear analog modulated unitary space-time matrix to a server. The server uses local PCs from multiple devices, which time-share the channel, to estimate the global PCs. The scheme  reduces  communication latency in three ways: 1) the coding and decoding is removed;  2)  CSI feedback is unnecessary since transmit CSI is not required; and 3) blind detection can avoid the need for explicit channel estimation. Like its digital counterpart, analog MIMO also uses antenna arrays to spatially multiplex data symbols but there is a distinctive difference. The symbols transmitted by the classic digital MIMO form parallel data streams, which are separately modulated, encoded and allocated different  transmission  power~\cite{Telatar1999}. Thereby, the channel capacity is achieved without  bit errors. In contrast, analog MIMO is an uncoded joint source-and-channel transmission  scheme customized for DPCA. The lack of coding exposes over-the-air signals to receiver noise, which, however, can be coped with in two    ways.  First, the server's  aggregation of local PC from multiple devices  suppresses not only data noise (i.e., deviation from the ground-truth) but also receiver noise, which diminishes as the number of devices grows~\cite{GXZhuonebit2020}. Second, the receiver noise can be suppressed by optimal PC estimation, which is a  main topic of this work.

The proposed analog MIMO for DPCA shares a common feature with two classic  MIMO schemes, namely non-coherent space-time modulation and analog MIMO channel feedback, in that they all involve transmission of a unitary matrix over a MIMO channel. Being a digital scheme, non-coherent space-time modulation is characterized by  a space-time constellation comprising  unitary matrices as its elements   \cite{Grassmannian2002TIT,NonTimDavidson2009,Marzetta2000Unitary, Marzetta2000Systematic}. The purpose of such a design is to support non-coherent communication where no CSI is required at the transmitter and receiver. This is made possible by the fact that the distortion  of a MIMO channel does not change the subspace represented by a  transmitted flat unitary matrix \cite{Marzetta2000Unitary}. The same principle is exploited in this work to enable blind detection. The capacity-achieving  constellation  design was found to be one that solves the packing problem on a Grassmann manifold, referring to a space of sub-spaces (or equivalently  unitary  matrices) \cite{Grassmannian2002TIT}. On the other hand, the transmitted unitary matrix in the scenario of  analog MIMO channel feedback is a pilot signal known by the receiver; its purpose is to assist the receiver to estimate the MIMO channel for forward-link transmission assuming the presence of channel duality \cite{CSI2006,Cairefeedback2010}. The scheme is found to be more efficient than digital CSI feedback, which quantizes and transmits CSI as bits, especially in a multi-user MIMO system where multi-user CSI feedback causes significant overhead~\cite{Cairefeedback2010}. Despite the above common feature, the current analog MIMO scheme differs from the other two in comparison.  Unlike analog MIMO channel feedback, the current analog MIMO is for data transmission instead of channel estimation. On the other hand, a unitary matrix transmitted by the analog MIMO is directly the payload data while that using the  non-coherent space-time modulation needs demodulation and decoding into bits that represent a quantized  version of the source data. 

We assume  time-sharing of the uplink channel between the devices and the server, and study both the cases with and without receive CSI. In principle, our techniques may be extended to the case when
multiple devices transmit simultaneously using over-the-air aggregation (AirAggregate) techniques that exploit the waveform superposition property of a multi-access channel (see, e.g.,~\cite{KYang2020,GXZhuonebit2020,GXZhuAirComp2019}).
This, however, would require efficient acquisition of transmit CSI and stringent synchronization across devices, which is difficult in our foreseen scenario. For example, we consider the DPCA traffic latency-critical and make no assumptions on uplink/downlink reciprocity (that could otherwise facilitate synchronization for AirAggregate purposes). The key component of the analog MIMO scheme is the global PC estimator that  is referred to simply as PC estimator. The designs associated with the cases with and without receive CSI are termed \emph{coherent} and \emph{blind} PC estimators, respectively. The main contributions made by this work are summarized as follows. 

\begin{itemize}

\item \emph{Optimal PC Estimation:} Given received local PCs, the optimal PC estimator is designed based on the \emph{maximum likelihood} (ML) criterion. First, the problem of optimal coherent PC estimation is solved in closed form by formulating it as a Procrustes problem. The resultant optimal estimator is characterized by \emph{regularized channel inversion} that balances the suppression  of data and receiver noise. On the other hand, to design the optimal blind PC estimator, the knowledge of channel distribution  is leveraged to transform the corresponding ML problem into one that is solved by finding the centroid of the  received local PCs by their projection onto a  Grassmannian. Theoretical analysis shows that both estimates are unbiased.

\item \emph{Estimation Error Analysis:} Measured by the  metric of \emph{mean square subspace distance} (MSSD) with respect to the ground-truth, the performance  of the preceding designs of optimal PC estimators is analyzed  by deriving upper bounds on their MSSD using manifold perturbation theory as a tool. The two bounds, which are tight as verified  by simulation, are observed to have the same form except for a difference by a multiplicative factor of two. As a result, the two estimators  are characterized by identical scaling laws, showing the dispensability  of receive CSI and the effectiveness of blind estimation  for DPCA. In particular, the  estimation error is \emph{inversely proportional} to the number of devices (or equivalently, the global   data size), achieving the same scaling as in the ideal case without receiver noise ~\cite{JF2019estimation2019,VC2020Arxiv}. Furthermore,  given a fixed number of PCs, the error is a monotone decreasing function of the data and channel \emph{signal-to-noise ratios} (SNRs) and the array sizes. It is observed that the spatial diversity gain contributed by arrays suppresses receiver noise more significantly  than data noise. As a result, when large arrays are deployed for analog MIMO, data noise becomes  the performance bottleneck  of DPCA. 

\end{itemize}

In addition, simulation results demonstrate that analog MIMO can achieve much shorter communication latency than its digital counterpart given similar error performance. Finally, it is also worth mentioning that the current design and analytical results can be extended to the case with transmit CSI where analog modulation facilitates over-the-air aggregation of local PCs simultaneously transmitted by devices, further reducing multi-access latency~\cite{GXZhuAirComp2019}. The extension is straightforward since channel inversion  at transmitters as required by the scheme removes MIMO channel distortion. 

The remainder of the paper is organized as follows. Section~\ref{Sec:ModelandMetrics} introduces system models and metrics. The designs and performance analysis of  coherent and blind global PC estimation  are presented in Section~\ref{Sec:coherent} and~\ref{Sec:blind}, respectively, assuming the matching of PCs'  and transmit array's  dimensionality. The assumption is relaxed in Section~\ref{Sec:extension}. Numerical  results are provided  in Section~\ref{Sec:experiments}, followed by concluding remarks in Section~\ref{Sec:conclusion}.

%

\section{Models and Metrics}\label{Sec:ModelandMetrics}
As illustrated in Fig.~\ref{fig:noncoherentcommunication}, we consider a MIMO system for supporting DPCA over  $K$ edge devices and  coordinated by  an edge server.  Relevant models and metrics are described as follows. 

\begin{figure*} 
	\centering
	\begin{minipage}[b]{0.9\textwidth}
		\centering
		\includegraphics[width=\textwidth]{ 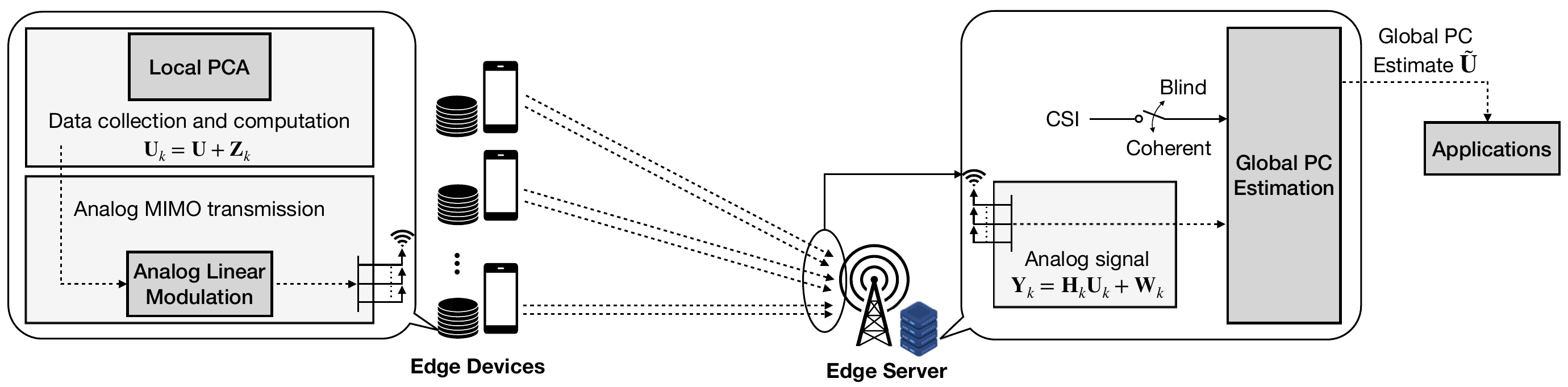}
		\vspace{-8mm}
	\end{minipage}
	\caption{DPCA enabled by analog MIMO communication.}
	\label{fig:noncoherentcommunication}
\end{figure*} 

\subsection{Distributed Principal Component Analysis}

To model DPCA, it is useful to first consider centralized PCA with \emph{independent and identically distributed} (i.i.d.) data samples $\mathbf{x}_{k,l}\in \mathbb{R}^{N\times 1}$ with $k=1,2,...,K$ and $l=1,2,...,L$. The PCs, an $M$-dimensional subspace, can be represented using a fat orthonormal matrix $\mathbf{U}\in \mathbb{R}^{M\times N}$ with $M < N$  since  the span of its row space,  $\mathsf{span}(\mathbf{U})$, specifies the PCs. Then centralized PCA is to find  $\mathbf{U}$ that provides the best representation of the data covariance by solving the following optimization problem: 
\begin{equation}
    \begin{aligned}\label{problem:PCA}
		\min_{\mathbf{U}} &\ \ \frac{1}{KL}\sum_{k=1}^K\sum_{l=1}^L \left\Vert \mathbf{x}_{k,l} -  \mathbf{U}^{\top}\mathbf{U}\mathbf{x}_{k,l}\right\Vert_2^2,\\
		\text{s.t.}&\ \ \mathbf{U}\mathbf{U}^{\top} = \mathbf{I}_M.
    \end{aligned}
\end{equation}
It is well known that the optimal solution, denoted as $\mathbf{U}^\star$,  for  \eqref{problem:PCA} is given by the $M$ dominant eigenvectors of the sample covariance matrix $\mathbf{R} = \frac{1}{KL}\sum_{k=1}^K\sum_{l=1}^{L}\mathbf{x}_{k,l}{\mathbf{x}}^{\top}_{k,l} $~\cite{gemp2021eigengame}. Some relevant, useful notation is introduced as follows.  Given an $M$-by-$M$ real positive semidefinite symmetric matrix $\mathbf{M}$ and its \emph{eigenvalue decomposition} $\mathbf{M} = \mathbf{Q}_{\mathbf{M}}\mathbf{\Sigma}_{\mathbf{M}}\mathbf{Q}_{\mathbf{M}}^{\top}$, the dominant  $M$-dimensional  eigenspace  can be represented by the orthonormal matrix 
$\mathcal{S}_M\left(\mathbf{M}\right)=\left[\mathbf{q}_1,\mathbf{q}_2,...,\mathbf{q}_M\right]^{\top}$,
where  $\mathbf{q}_i$ is the $i$-th column of $\mathbf{Q}_{\mathbf{M}}$. Using the notation allows the PCs to be written as $\mathcal{S}_M\left(\mathbf{R}\right)$. 

Next, DPCA differs  from its centralized counterpart mainly in that the $KL$ samples are uniformly distributed over $K$ edge devices. The local PCs at each  device, say $\mathbf{U}_k$ at device $k$,  are computed also using eigenvalue decomposition but based on the local dataset. Due to the reduced dataset, $\mathbf{U}_k$ deviates from the ground-truth $\mathbf{U}^\star$ and the deviation is modeled as additive noise, termed \emph{data noise}:  
$\mathbf{U}_k = \mathbf{U}^\star + \mathbf{Z}_k$, 
where  $\{\mathbf{Z}_k\}$, which  model data noise,  are i.i.d. Gaussian $\mathcal{N}\left(0,\sigma_{\text{d}}^2\right)$ random variables. The data-noise model  yields theoretic tractability and is commonly adopted both in the  distributed PCA   (see, e.g.,~\cite{Calderbank2015,grammenos2020federated}) and  the  distributed learning literature  (see, e.g.,~\cite{signSGD2018}). With this model, the noise  variance models the effect of the limited   local date size.  The ratio between the local and global data size is termed  the local dataset reduction factor; an increase of this ratio corresponds to  a growing data noise variance, and vice versa. The data SNR, denoted as $\gamma_{\text{d}}$,  is defined as the ratio between the \textit{average power} of desired data $\mathbf{U}$, i.e. $\mathsf{E}\left[\Vert \mathbf{U} \Vert_F^2\right]/N=M/N$, and the data noise variance: 
\begin{equation}\label{eq:datanoise}
	\gamma_{\text{d}}=\frac{M}{N\sigma_{\text{d}}^2}.
\end{equation}

\subsection{Modelling Analog MIMO Communication}
First, the transmission  and channel models for analog MIMO communication are described as follows. Let $N_\text{r}$ and $N_{\text{t}}$ with $N_{\text{r}} \geq N_{\text{t}}$ denote the numbers of antennas at the edge server and each device, respectively. The path losses of different devices are assumed to have been compensated for by power control at the base station, such that all devices have the same effective average path gain. Herein, we consider a standard Rayleigh fading model, where the MIMO channel coefficients are represented by a complex random  matrix  $\mathbf{H}_k\in \mathbb{C}^{N_{\text{r}}\times N_{\text{t}}}$ with i.i.d.  $\mathcal{CN}(0,1)$ entries. For other fading models, for example  Rician fading, the effective channel rank is likely less than in the  Rayleigh fading model. In this case, a matrix symbol may need to be partitioned and transmitted over multiple symbol durations following the procedure in Section~\ref{Sec:extension}; the proposed estimator designs, however, remain applicable. For ease of notation and exposition, $N_{\text{t}}$ is assumed to be identical to  the PC dimensions,  $M$; this assumption is relaxed in Section~\ref{Sec:extension}. 

We further consider a \emph{matrix symbol} as a data unit for transmission, which is defined as follows. In each time slot,  an $N_t\times 1$ vector of complex scalar symbols is transmitted over $N_t$ antennas. Then a (matrix) symbol of duration $N$ slots supports the transmission of an  $N_t\times N$ matrix. Thereby, an orthonormal matrix representing the local PCs computed at a device,  say $\{\mathbf{U}_k\}$ at device $k$, is transmitted using a single matrix symbol duration. Devices time-share the channel for transmission. It is assumed that the channel coherence time is longer than a matrix symbol duration, such that the channel is  constant within each symbol duration  but may vary between different symbols.

Consider the transmission of an arbitrary device, say device $k$. The devices have no  CSI, and the  (matrix) symbol, $\mathbf{U}_k$,  is  transmitted using uncoded linear analog modulation as in non-coherent MIMO \cite{Grassmannian2002TIT}. The transmission of  $\mathbf{U}_k$ leads to the server receiving the  following matrix $\mathbf{Y}_k $:
\begin{equation}\label{eq:model}
	\mathbf{Y}_k = \mathbf{H}_k\mathbf{U}_k + \mathbf{W}_k = \mathbf{H}_k\left(\mathbf{U} + \mathbf{Z}_k\right) + \mathbf{W}_k,
\end{equation} 
where $\mathbf{W}_k\in \mathbb{C}^{N_{\text{r}}\times N}$ indicates receiver noise with i.i.d. complex Gaussian $\mathcal{CN}\left(0,\sigma_{\text{c}}^2\right)$ elements. The average power of each transmitted matrix symbol is $P = \mathsf{E}[\Vert \mathbf{U}_k \Vert_F^2]/N= \left(M + M^2/\gamma_{\text{d}}\right)/N$. Then let $\gamma_{\text{c}}$ denote the channel SNR defined as
\begin{equation}\label{eq:channelnoise}
	\gamma_{\text{c}} = \frac{P}{\sigma_\text{c}^2}= \frac{M + M^2/{\gamma_{\text{d}}}}{N\sigma_\text{c}^2}.
\end{equation}
Since each  matrix symbol is real but the channel is complex, the server obtains two observations of the transmitted symbol   from each received symbol. It is useful to separate the real and imaginary parts of the received symbol, $\mathbf{Y}_k$, denoted as  $\mathbf{Y}_k^{\text{Re}}$ and $\mathbf{Y}_k^{\text{Im}}$, respectively. They can be written  as 
\begin{align}
    \mathbf{Y}_k^{\text{Re}} 
    &= \Re\{\mathbf{Y}_k\} = \mathbf{H}_k^{\text{Re}} \left(\mathbf{U} + \mathbf{Z}_k\right) + \mathbf{W}_k^{\text{Re}},\nonumber\\
    \mathbf{Y}_k^{\text{Im}} 
    &= \Im\{\mathbf{Y}_k\} = \mathbf{H}_k^{\text{Im}} \left(\mathbf{U} + \mathbf{Z}_k\right) + \mathbf{W}_k^{\text{Im}}\nonumber,
\end{align}
where $\mathbf{H}_k^{\text{Re}}$ and $\mathbf{H}_k^{\text{Im}}$ represent  the real and imaginary parts of the channel and are both distributed with  i.i.d. Gaussian $\mathcal{N}(0,1/2)$ entries, and $\mathbf{W}_k^{\text{Re}}$ and $\mathbf{W}_k^{\text{Im}}$ are similarly defined with i.i.d. Gaussian $\mathcal{N}(0,\sigma_{\text{c}}^2/2)$ entries. 
Using this notation allows the following alternative expressions of the received matrix symbol: 
\begin{align}\label{eq:newModel}
	\hat{\mathbf{Y}}_k 
	&= \sqrt{2}
		\begin{bmatrix}
			\mathbf{Y}_k^{\text{Re}}\\
			\mathbf{Y}_k^{\text{Im}}
		\end{bmatrix},\nonumber\\
	& = \sqrt{2}
		\begin{bmatrix}
			\mathbf{H}_k^{\text{Re}}\\
			\mathbf{H}_k^{\text{Im}}
		\end{bmatrix} 
		\left(\mathbf{U} + \mathbf{Z}_k\right) + \sqrt{2}
		\begin{bmatrix}
			\mathbf{W}_k^{\text{Re}}\\
			\mathbf{W}_k^{\text{Im}}
		\end{bmatrix},\nonumber\\
	& \overset{\triangle}{=}\hat{\mathbf{H}}_k\left(\mathbf{U} + \mathbf{Z}_k\right) + \hat{\mathbf{W}}_k. 
\end{align}

 Next, consider the PC estimation at the receiver. For the case with receive CSI, the receiver is assumed to have perfect   knowledge of  the channel matrix $\mathbf{H}_k$ and channel SNR $\gamma_{\text{c}}$ as well as of the data SNR $\gamma_{\text{d}}$. Such knowledge is not required in the blind case. Both the coherent PC estimator in the former case and the blind one in the latter are designed using the ML criterion. Let  $\mathcal{L}(\mathbf{U};\hat{\mathbf{Y}}_k)$ denote  the likelihood function of the estimate $\mathbf{U}$ given the observations in $\hat{\mathbf{Y}}_k$. Then the PC estimation can be formulated as  
\begin{equation}
    \begin{aligned}\label{problem:ML}
    	\max_{\mathbf{U}} &\ \ \mathcal{L}(\mathbf{U};\hat{\mathbf{Y}}_k),\\
    	\text{s.t.}&\ \ \mathbf{U}\mathbf{U}^{\top} = \mathbf{I}_M.
    \end{aligned}
\end{equation}
The problem is solved in the following sections   to design coherent and blind PC estimators. A key property that underpins both the    current analog MIMO as well as  non-coherent MIMO   is that the subspace represented by a transmitted matrix, say $\mathbf{U}_k$, is invariant to a MIMO channel rotation, which can be mathematically described  as~\cite{Yuqing2018AutomaticRecognition}
\begin{equation}\label{eq:rotation}
	\mathcal{S}_M\left(\mathbf{U}_k^{\top}\mathbf{H}_k^{\top}\mathbf{H}_k\mathbf{U}_k\right)=\mathcal{S}_M\left(\mathbf{U}_k^{\top}\mathbf{U}_k\right).
\end{equation}

\subsection{PC Error Metric}\label{subsection:definition}

The PC estimation  aims at estimating a subspace containing the ground-truth PCs but their particular orientations within the subspace is irrelevant. In view of this fact, it is most convenient to consider a Grassmann manifold, also called a Grassmannian and referring to a space of subspaces, since a subspace appears as a single point in the manifold~\cite{Yuqing2018AutomaticRecognition}.   Let  $\mathcal{G}_{N,M}$ denote a Grassmannian comprising $M$-dimensional subspaces embedded in an $N$-dimensional space. For convenience, a point in $\mathcal{G}_{N,M}$ is commonly represented using an orthonormal matrix, say $\mathbf{U}$. However, it should be clarified that the point more precisely  corresponds to $\mathsf{span}\left(\mathbf{U}\right)$ and hence the set of subspaces, $\{\mathbf{U}'\}$, which satisfy the equality,  $\mathsf{span}(\mathbf{U})=\mathsf{span}\left(\mathbf{Q}\mathbf{U}'\right)$,  for some $M\times M$ orthonormal matrix $\mathbf{Q}$. It follows from the above discussion that a suitable PC distortion metric should measure the distance between two points in the Grassmannian, which correspond to the subspaces of the estimated PCs and their ground-truth. Among many available subspace distance measures, 
the Euclidean subspace distance,  which  is adopted in this work,  is a popular choice for its tractability \cite{JF2019estimation2019,Yuqing2018AutomaticRecognition}. Using this metric and with the estimated PCs represented by the orthonormal matrix $\tilde{\mathbf{U}}$, its  distance to the ground-truth $\mathbf{U}$ can be defined as 
\begin{equation}\label{eq:distance}
	d(\mathbf{U},\tilde{\mathbf{U}}) = \left\Vert \tilde{\mathbf{U}}^{\top}\tilde{\mathbf{U}} - \mathbf{U}^{\top}\mathbf{U}\right\Vert_F.
\end{equation}
Its geometric meaning is reflected by the  alternative expression  in terms of the principal angles between the subspaces  spanned by $\tilde{\mathbf{U}}$ and $\mathbf{U}$, denoted as $\{\theta_i\}_{i=1}^M$. Define the diagonal matrix $\mathbf{\Theta}(\mathbf{U},\tilde{\mathbf{U}}) = \mathsf{diag}\left(\theta_1,\theta_2,...,\theta_M\right)$ for those angles;  then $d^2(\mathbf{U},\tilde{\mathbf{U}}) =2\sum_{i=1}^M\mathop{\sin}^2(\theta_i)$~\cite{JF2019estimation2019}. Building on the subspace distance measure, we define the PC error metric as follows.  

\begin{Definition}[Mean square subspace distance]\emph{
The \emph{mean square subspace distance} (MSSD) between  $\mathbf{U}$ and $\tilde{\mathbf{U}}$ is defined as $d_{\text{ms}}(\mathbf{U},\tilde{\mathbf{U}})=\mathsf{E}\left[ d^2(\mathbf{U},\tilde{\mathbf{U}}) \right]$, i.e.
\begin{equation} \label{eq:MSSD}
	d_{\text{ms}}(\mathbf{U},\tilde{\mathbf{U}}) = 2M-2\mathsf{E}\left[\mathsf{Tr}(\tilde{\mathbf{U}}^{\top}\tilde{\mathbf{U}}\mathbf{U}^{\top}\mathbf{U})\right].
\end{equation}
where the expectation is taken over the distributions of  data noise $\mathbf{Z}$, receiver noise $\mathbf{W}$, and also channel gain $\mathbf{H}$ in the case of blind PC estimation. 
}
\end{Definition}

%

\section{PC Estimation for Coherent MIMO Receivers} \label{Sec:coherent}

\subsection{Design of Coherent PC Estimator}
First, the usefulness of receive CSI is to facilitate coherent combination of the received data. To this end, channel inversion is applied to a received matrix symbol, say $\hat{\mathbf{Y}}_k$,  by multiplication with the zero-forcing matrix,   $\hat{\mathbf{H}}_k^+ \overset{\triangle}{=}(\hat{\mathbf{H}}_k^{\top}\hat{\mathbf{H}}_k)^{-1}\hat{\mathbf{H}}_k^{\top}$, which is the pseudo inverse of the channel $\hat{\mathbf{H}}_k$. Using \eqref{eq:newModel}, this yields the following observations for PC estimation: $\hat{\mathbf{H}}_k^+\hat{\mathbf{Y}}_k = \mathbf{U} + \mathbf{Z}_k + \hat{\mathbf{H}}_k^+\hat{\mathbf{W}}_k, \forall k$.
The processed observations follow a matrix Gaussian distribution $\mathcal{MN}\left(\mathbf{U},\mathbf{\Sigma}_k,\mathbf{I}_N\right)$ with $\mathbf{\Sigma}_k = \sigma_{\text{d}}^2\mathbf{I}_M+\sigma_{\text{c}}^2{\hat{\mathbf{H}}_k^+}\left({\hat{\mathbf{H}}_k^+}\right)^{\top}$. Specifically, 
\begin{align} 
    &p(\hat{\mathbf{H}}_k^+\hat{\mathbf{Y}}_k|\mathbf{U},\mathbf{H}) \nonumber\\
    &=\frac{\exp\left(-\frac{1}{2}\mathsf{Tr}\left((\hat{\mathbf{H}}_k^+\hat{\mathbf{Y}}_k-\mathbf{U})^{\top}\mathbf{\Sigma}_k^{-1}(\hat{\mathbf{H}}_k^+\hat{\mathbf{Y}}_k-\mathbf{U})\right)\right)}{(2\pi)^{\frac{MN}{2}}\mathsf{det}\left(\mathbf{\Sigma}_k\right)^{\frac{N}{2}}} \nonumber.
\end{align}
Next, the ML PC estimator is derived by maximizing the likelihood function that is defined as the joint probability of the observations $\{\hat{\mathbf{H}}_k^+\hat{\mathbf{Y}}_k\}_{k=1}^K$ conditioned on the ground-truth $\mathbf{U}$ and channel realizations $\{\hat{\mathbf{H}}_k\}_{k=1}^K$. Since  $\{\hat{\mathbf{H}}_k^+\hat{\mathbf{Y}}_k\}_{k=1}^K$ are mutually independent, the logarithm of the likelihood function can be written as
\begin{align}
	&\mathcal{L}(\mathbf{U;\hat{\mathbf{Y}}},\hat{\mathbf{H}})\nonumber\\
	&= \ln\left(\prod_{k=1}^{K} p(\hat{\mathbf{H}}_k^+\hat{\mathbf{Y}}_k|\mathbf{U},\hat{\mathbf{H}}_k)\right), \nonumber\\ 
	&= -\frac{1}{2}\sum_{k=1}^K\mathsf{Tr}\left((\hat{\mathbf{H}}_k^+\hat{\mathbf{Y}}_k-\mathbf{U})^{\top}{\mathbf{\Sigma}_k}^{-1}(\hat{\mathbf{H}}_k^+\hat{\mathbf{Y}}_k-\mathbf{U})\right) \nonumber \\
	&\,\quad - \frac{1}{2}MNK\ln(2\pi) - \frac{N}{2}\sum_{k=1}^{K}\ln\mathsf{det}\left(\mathbf{\Sigma}_k\right) \nonumber.
\end{align}
One can observe from the above expression that the variable $\mathbf{U}$ only enters into the first term which is further expanded as: 
\begin{align}
	&\sum_{k=1}^K\mathsf{Tr}\left(\mathbf{U}^{\top}\mathbf{\Sigma}_k^{-1}\hat{\mathbf{H}}_k^+\hat{\mathbf{Y}}_k\right) \nonumber\\
	& - \frac{1}{2}\sum_{k=1}^K\mathsf{Tr}\left(\hat{\mathbf{Y}}_k^{\top}(\hat{\mathbf{H}}_k^+)^{\top}\mathbf{\Sigma}_k^{-1}\hat{\mathbf{H}}_k^+\hat{\mathbf{Y}}_k + \mathbf{\Sigma}_k^{-1}\right) \nonumber.
\end{align}
It follows that the likelihood function is determined by $\mathbf{U}$ through $\sum_{k=1}^K\mathsf{Tr}\left(\mathbf{U}^{\top}\mathbf{\Sigma}_k^{-1}\hat{\mathbf{H}}_k^+\hat{\mathbf{Y}}_k\right)$. 

This allows the ML problem in \eqref{problem:ML} to be particularized for the current case as 
\begin{equation}
    \begin{aligned}\label{eq:procrustes}
    	\mathop{\max}_{\mathbf{U}}&\ \ \frac{1}{K}\sum_{k=1}^K\mathsf{Tr}\left(\mathbf{U}^{\top}\mathbf{\Sigma}_k^{-1}\hat{\mathbf{H}}_k^+\hat{\mathbf{Y}}_k\right), \\
    	\mathrm{s.t.}&\ \ \mathbf{U}\mathbf{U}^{\top}=\mathbf{I}_M.
    \end{aligned}
\end{equation}
The problem in  \eqref{eq:procrustes} is non-convex due to the feasible region restricted to a hyper-sphere. One solution method is to transform it into an   equivalent, tractable  orthogonal Procrustes problem~\cite{JCGower2004}. To this end, consider the following SVD
\begin{equation} \label{eq:summation}
	\mathbf{J} = \frac{1}{K}\sum_{k=1}^K\mathbf{\Sigma}_k^{-1}\hat{\mathbf{H}}_k^+\hat{\mathbf{Y}}_k = \mathbf{U}_{\mathbf{J}}\mathbf{\Lambda}_{\mathbf{J}}\mathbf{V}_{\mathbf{J}}^{\top},
\end{equation}
where the diagonal matrix $\mathbf{\Lambda}_{\mathbf{J}}=\mathsf{diag}\left(\lambda_{\mathbf{J},1},...,\lambda_{\mathbf{J},M}\right)$ with $\lambda_{\mathbf{J},m}$ representing the $m$-th singular value, and  $\mathbf{U}_{\mathbf{J}}$ and $\mathbf{V}_{\mathbf{J}}^{\top}$ are orthonormal eigen matrices. Let them be expressed in terms of the singular vectors of $\mathbf{J}$: $\mathbf{U}_{\mathbf{J}}=\left[\mathbf{u}_{\mathbf{J},1},...,\mathbf{u}_{\mathbf{J},M}\right]$, $\mathbf{V}_{\mathbf{J}}^{\top}=\left[\mathbf{v}_{\mathbf{J},1},...,\mathbf{v}_{\mathbf{J},M}\right]^{\top}$.  Let $\tilde{\mathbf{U}}^{\star}$ denote the optimal solution for  the problem in \eqref{eq:procrustes}, which yields  the optimal estimate of the ground-truth $\mathbf{U}$. 

Then substituting \eqref{eq:summation} into the objective function in \eqref{eq:procrustes} allows $\tilde{\mathbf{U}}^{\star}$ to be given as 
\begin{align}
	\tilde{\mathbf{U}}^{\star}
	& = \mathop{\arg\max}_{\mathbf{U}\in \mathcal{G}_{N,M}}\; \mathsf{Tr}\left(\mathbf{U}^{\top}\mathbf{J}\right), \nonumber\\
	& = \mathop{\arg\max}_{\mathbf{U}\in \mathcal{G}_{N,M}}\; \sum_{m=1}^M \lambda_{\mathbf{J},m}\cdot \left\langle \mathbf{U}^{\top}\mathbf{u}_{\mathbf{J},m},\mathbf{v}_{\mathbf{J},m} \right\rangle \nonumber.
\end{align}
The above problem can be solved using the fact that  $\sum_{m=1}^M \lambda_{\mathbf{J},m}\cdot \left\langle \mathbf{U}^{\top}\mathbf{u}_{\mathbf{J},m},\mathbf{v}_{\mathbf{J},m}\right\rangle \leq \sum_{m=1}^M \lambda_{\mathbf{J},m}$ with equality when  $\mathbf{U}^{\top}\mathbf{u}_{\mathbf{J},m} = \mathbf{v}_{\mathbf{J},m}, \forall m$. It follows that  $\tilde{\mathbf{U}}^{*}{}^{\top}\mathbf{U}_{\mathbf{J}} = \mathbf{V}_{\mathbf{J}}$, or equivalently $\tilde{\mathbf{U}}^{\star}  = \mathbf{U}_{\mathbf{J}}\mathbf{V}_{\mathbf{J}}^{\top}$. Using~\cite[Theorem~7.3.1]{MatrixAnalysis},  $\mathbf{U}_{\mathbf{J}}\mathbf{V}_{\mathbf{J}}^{\top}$ can be obtained from the \emph{polar decomposition} of the matrix summation $ \mathbf{J} $ in \eqref{eq:summation}, i.e.,  $\mathbf{U}_{\mathbf{J}}\mathbf{V}_{\mathbf{J}}^{\top} = (\mathbf{J}{\mathbf{J}}^{\top})^{-1/2}\mathbf{J}$, and furthermore spans the same principal eigenspace as $\mathsf{span}\left(\mathbf{V}_{\mathbf{J}}^{\top}\right)$. This leads to the following main result of this sub-section. 
\begin{Theorem}[Optimal PC Estimation  with Receive CSI]\label{Theorem:coherent}
\emph{Given the channel matrices $\{\hat{\mathbf{H}}_k\}_{k=1}^K$ and the received matrix symbols  $\{\hat{\mathbf{Y}}_k\}_{k=1}^K$,  the optimal global PCs based on the  ML criterion are given as 
\begin{equation}\label{eq:coherent}
	\tilde{\mathbf{U}}^{\star} = \mathcal{S}_M\left( {\mathbf{J}}^{\top}\mathbf{J}\right),
\end{equation}
where $\mathbf{J}=\frac{1}{K}\sum_{k=1}^{K} \left[\sigma_{\text{d}}^2\mathbf{I}_M+\sigma_{\text{c}}^2\hat{\mathbf{H}}_k^+(\hat{\mathbf{H}}_k^+)^{\top}\right]^{-1}\hat{\mathbf{H}}_k^+\hat{\mathbf{Y}}_k$.
}
\end{Theorem}
Theorem~\ref{Theorem:coherent} suggests  that the optimal  coherent PC estimator should 1) first coherently combine  received observations  with weights reflecting a  linear \emph{minimum-mean-square error} (MMSE) receiver  and 2) then compute the dominant singular space of the result, yielding the optimal estimate of the ground-truth PCs. The weights that aim at coping with MIMO channels and noise  distinguish the current DPCA over wireless channels from the  conventional designs where weights are either unit~\cite{JF2019estimation2019} or independent of channels~\cite{pmlr-v119-huang20e,VC2020Arxiv}. Furthermore, the summation form of $\mathbf{J}$ suggests that when transmit CSI is available, the proposed coherent PC estimator can be further developed to enable over-the-air aggregation exploiting the waveform superposition property of multi-access channels~\cite{GXAirCompMag2021}. To be specific, transmit and beamforming matrices, denoted by $\{\mathbf{B}_{\mathrm{t},k}\}$ and $\mathbf{B}_{\mathrm{r}}$ respectively, can be jointly designed to satisfy $\mathbf{B}_{\mathrm{r}}\hat{\mathbf{H}}_k\mathbf{B}_{\mathrm{t},k} = \left[\sigma_{\text{d}}^2\mathbf{I}_M+\sigma_{\text{c}}^2\hat{\mathbf{H}}_k^+(\hat{\mathbf{H}}_k^+)^{\top}\right]^{-1}$ and
\begin{align*}
    \mathbf{B}_{\mathrm{r}}\mathbf{B}_{\mathrm{r}}^{\top}=\frac{1}{K^2}\sum_{k=1}^{K}& \left[\sigma_{\text{d}}^2\mathbf{I}_M+\sigma_{\text{c}}^2\hat{\mathbf{H}}_k^+(\hat{\mathbf{H}}_k^+)^{\top}\right]^{-1}\hat{\mathbf{H}}_k^{+}\\
    &\cdot\left(\left[\sigma_{\text{d}}^2\mathbf{I}_M+\sigma_{\text{c}}^2\hat{\mathbf{H}}_k^+(\hat{\mathbf{H}}_k^+)^{\top}\right]^{-1}\hat{\mathbf{H}}_k^{+}\right)^{\top}
\end{align*}
such that the superposition of signals from different edge devices is equal to $\mathbf{J}$.
\begin{Remark}[Data-and-Receiver Noise Regularization]
\emph{As shown in Theorem~\ref{Theorem:coherent},  the zero-forcing matrix $\hat{\mathbf{H}}_k^+$ is used to equalize the channel distortion but by doing so, it may amplify  receiver noise given a poorly conditioned channel. This issue is addressed by a regularization term, $\left[\sigma_{\text{d}}^2\mathbf{I}_M+\sigma_{\text{c}}^2\hat{\mathbf{H}}_k^+(\hat{\mathbf{H}}_k^+)^{\top}\right]$, that balances data and receiver noise based on the knowledge of their covariance matrices.}
\end{Remark}
The complexity of the coherent estimator is analyzed as follows. First, the matrix multiplication, $\mathbf{J^{\top}\mathbf{J}}$, and its eigenvalue decomposition have the combined  complexity $O(2M^3)$. Then considering  $\mathbf{J}$, the  terms $\{\left[\sigma_{\text{d}}^2\mathbf{I}_M+\sigma_{\text{c}}^2\hat{\mathbf{H}}_k^+(\hat{\mathbf{H}}_k^+)^{\top}\right]^{-1}\hat{\mathbf{H}}_k^+\hat{\mathbf{Y}}_k\}$ should be computed separately by matrix multiplication and inversion, which leads to  complexity of $O(4KM^2N_{\mathrm{r}}+2KM^3+2KMN_{\mathrm{r}}N)$. Combining these results, the overall complexity of the coherent estimator is $O(2KM^2N_{\mathrm{r}}+KM^3+MKN_{\mathrm{r}}N+M^3)$.

Last, the PC estimation in  Theorem~\ref{Theorem:coherent} is unbiased as shown below. 

\begin{Corollary} \label{Pro:unbiasedEwCSI}
	\emph{The  coherent PC estimator in Theorem~\ref{Theorem:coherent} together with the preceding channel inversion achieves unbiased estimation of the global PCs in the following sense. Define  $\bar{\mathbf{J}} = \mathsf{E}\left[\mathbf{J}\right]$ with $\mathbf{J}$ given by Theorem~\ref{Theorem:coherent}, and the optimal PC estimate $\bar{\mathbf{U}}^\star = \mathcal{S}_M\left( {\bar{\mathbf{J}}}^{\top}\bar{\mathbf{J}}\right)$. Then $d\left(\bar{\mathbf{U}}^\star,\mathbf{U}\right) = 0$.}
\end{Corollary}

\begin{proof}According to \eqref{eq:summation}, the $\bar{\mathbf{J}}$ is given by
\begin{align}
	\bar{\mathbf{J}}
	& =\mathsf{E}\left[\frac{1}{K}\sum_{k=1}^{K} [\sigma_{\text{d}}^2\mathbf{I}_M+\sigma_{\text{c}}^2\hat{\mathbf{H}}_k^+(\hat{\mathbf{H}}_k^{+})^{\top}]^{-1}\hat{\mathbf{H}}_k^+\hat{\mathbf{Y}}_k\right],\nonumber\\
	& = \mathsf{E}\left[[\sigma_{\text{d}}^2\mathbf{I}_M+\sigma_{\text{c}}^2\hat{\mathbf{H}}_k^+(\hat{\mathbf{H}}_k^{+})^{\top}]^{-1}\right]\mathbf{U},\nonumber.
\end{align}
Define $\mathbf{C}\overset{\triangle}{=}\mathsf{E}\left[[\sigma_{\text{d}}^2\mathbf{I}_M+\sigma_{\text{c}}^2\hat{\mathbf{H}}_k^+(\hat{\mathbf{H}}_k^{+})^{\top}]^{-1}\right]$ and clearly $\mathsf{rank}\left(\mathbf{C}\right)=M$. Then, based on the property in~\eqref{eq:rotation}, we have $\mathcal{S}_M\left(\mathbf{U}^{\top}\mathbf{C}^{\top}\mathbf{C}\mathbf{U}\right) = \mathcal{S}_M\left(\mathbf{U}^{\top}\mathbf{U}\right)$ which completes the proof.
\end{proof} 
The above result  also implies that as the number of devices, $K$, increases (so does the global data size), the optimal estimate $\tilde{\mathbf{U}}^{\star}$ in Theorem~\ref{Theorem:coherent} converges to the ground-truth since  $\mathbf{J} \rightarrow \bar{\mathbf{J}}$. 

\subsection{Error Performance Analysis}\label{sec:coherentunbiasedness}

The performance of the optimal  coherent PC estimation in the preceding sub-section can be analyzed using a result from  perturbation theory for  Grassmann manifolds. 
To this end, consider two perturbed subspaces  $\mathbf{V}_1 = \mathcal{S}_M\left(\mathbf{G}_1^{\top}\mathbf{G}_1\right)$ and $\mathbf{V}_2 = \mathcal{S}_M\left(\mathbf{G}_2^{\top}\mathbf{G}_2\right)$, where $\mathbf{G}_1 = \mathbf{F} + \epsilon\mathbf{E}_1$ and $\mathbf{G}_2 = \mathbf{F} + \epsilon\mathbf{E}_2$ with the orthonormal matrix $\mathbf{F} \in \mathcal{G}_{N,M}$ representing a ground-truth, $\epsilon\mathbf{E}_i, i=1,2$ denoting random additive perturbations, and $\epsilon > 0$ controlling the perturbation magnitude. The following result is from~\cite[Theorem~9.3.4]{YChikuse2012}. 
\begin{Lemma}\label{Lemma:subspaceDistanceBetweenTwoPerturbedSubspace}\emph{
The squared subspace distance of the  perturbed subspaces  $\mathbf{V}_{1} $ and $ \mathbf{V}_{2}$ satisfies 
\begin{equation}
	d^2(\mathbf{V}_{1},\mathbf{V}_{2}) = 2\epsilon^2\mathsf{Tr}( \mathbf{\Delta_{\mathbf{E}}}\mathbf{F^{\bot}}^{\top}\mathbf{F}^{\bot}\mathbf{\Delta_{\mathbf{E}}}^{\top}) + O\left(\epsilon^3\right),
\end{equation}
where $\mathbf{F^{\bot}} $ is the orthogonal complement of $ \mathbf{F}$ and $\mathbf{\Delta_{\mathbf{E}}} = \mathbf{E}_1 - \mathbf{E}_2$.}
\end{Lemma}
 Using Lemma~\ref{Lemma:subspaceDistanceBetweenTwoPerturbedSubspace}, it is proved in Appendix~\ref{proof:EwCSIdistance} that the PC estimation error can be characterized as follows. 
\begin{Lemma}[]\label{Lemma:EwCSIdistance}\emph{
The square error of  the optimal PC estimate $\tilde{\mathbf{U}}^{\star}$ in Theorem~\ref{Theorem:coherent} is
\begin{equation}
	d^2(\tilde{\mathbf{U}}^{\star},\mathbf{U}) = 2\mathsf{Tr}( \mathbf{\Delta}_{\mathbf{E}}\mathbf{U^{\bot}}^{\top}\mathbf{U}^{\bot}{\mathbf{\Delta}_{\mathbf{E}}}^{\top}) + O\left(\Vert\mathbf{\Sigma}_{\mathbf{E}}^{-1}\Vert_F^{3}\right),
\end{equation}
where $\mathbf{\Delta}_{\mathbf{E}}\sim\mathcal{MN}\left[\mathbf{0},\mathbf{\Sigma}_{\mathbf{E}}^{-1} ,\mathbf{I}_N\right]$ with $\mathbf{\Sigma}_{\mathbf{E}} = \sum_{k=1}^{K} (\sigma_{\text{d}}^2\mathbf{I}_M+\sigma_{\text{c}}^2\hat{\mathbf{H}}_k^+(\hat{\mathbf{H}}_k^+)^{\top})^{-1}$ and $\mathbf{U^{\bot}} $ is the orthogonal complement of $ \mathbf{U}$.}
\end{Lemma}
Based on the definition of $\mathbf{\Sigma}_{\mathbf{E}}$, one can infer that the residual term  $O\left(\Vert\mathbf{\Sigma}_{\mathbf{E}}^{-1}\Vert_F^{3}\right)$ in Lemma~\ref{Lemma:EwCSIdistance} diminishes as  the number of devices  (or the global data size) $K\rightarrow \infty$ or data noise and receiver noise $\sigma_{\text{d}},\sigma_{\text{c}}\rightarrow 0$. Then using Lemma~\ref{Lemma:EwCSIdistance} with the residual term omitted, the error of optimal PC estimation can be characterized as follows. 

\begin{Theorem}[Error of Optimal Coherent PC Estimation]\label{The:upperBoundMSSDEwCSI}\emph{
Given many devices ($K\rightarrow \infty$) and small channel and data noise ($\sigma_{\text{d}},\sigma_{\text{c}}\rightarrow 0$), the MSSD of the optimal coherent PC estimator in Theorem~\ref{Theorem:coherent} is asymptotically upper bounded as
\begin{equation}\label{eq:coherentupbound1}
	d_{\text{ms}}(\tilde{\mathbf{U}}^{\star},\mathbf{U}) \leq \frac{2M(N-M)}{K}\left(\sigma_{\text{d}}^2 + \frac{\sigma_{\text{c}}^2}{2N_{\text{r}}-M-1}\right), 
\end{equation}
for a MIMO system (i.e., $2N_r -M > 1$). 
}
\end{Theorem}
\begin{proof}
See Appendix~\ref{proof:upperBoundMSSDEwCSI}.
\end{proof}

Note that as $M$ increases, the first term in the MSSD upper bound in~\eqref{eq:coherentupbound1}, $\frac{2M(N-M)}{K}$, grows  if $M\leq \frac{N}{2}$ but otherwise decreases. This is due to the following  well known property of a random point (subspace) uniformly distributed on a Grassmannian.  Its uncertainty (and hence  estimation error) grows with its dimensionality $M$, reaches the maximum at $M=\frac{N}{2}$, and reduces afterwards since it can be equivalently represented by a lower-dimensional complementary subspace~\cite{Grassmannian2002TIT}. Next, the MSSD upper bound can be rewritten in terms of the channel and data SNRs given in ~\eqref{eq:datanoise} and~\eqref{eq:channelnoise} as 
\begin{equation}\label{eq:coherentupbound2}
	d_{\text{ms}}(\tilde{\mathbf{U}}^{\star},\mathbf{U}) \leq \frac{2M^2(N-M)}{KN}\left[\gamma_{\text{d}}^{-1} + \frac{(1+M\gamma_\text{d}^{-1})\gamma_{\text{c}}^{-1}}{2N_{\text{r}}-M-1}\right] .
\end{equation}
It is observed from~\eqref{eq:coherentupbound2} that given fixed data and channel SNRs, the estimation error  increases with  $M$ and reaches a peak at  $M=2N/3$ instead of $M = N/2$ when noise variances are fixed. 

Based on Theorem~\ref{The:upperBoundMSSDEwCSI}, the effects of different system parameters on the error performance of coherent PC estimation are described as follows. 

\begin{itemize}
    \item The MSSD linearly decays as the device number $K$, or equivalently the number of  noisy observations of the global PCs, increases. The scaling law is identical to that for conventional DPCA over  reliable links (see, e.g.~\cite{JF2019estimation2019} and~\cite{VC2020Arxiv}). 
    The result suggests increasingly accurate estimation of the ground-truth global PCs as more devices participate in DPCA.
    
    \item The result in~\eqref{eq:coherentupbound2} shows a trade-off between the data and receiver noise. Furthermore, if receiver noise is negligible, one can observe from~\eqref{eq:coherentupbound2} that the estimation error diminishes about \emph{inversely} with increasing data SNR, namely  $O\left(\frac{2M^2}{K\gamma_{\text{d}}}\right)$ if $N\gg M$. On the other hand, if receiver noise is dominant, the error decays with increasing channel SNR as 
    $O\left(\frac{1}{(2N_{\text{r}}-N_t-1)\gamma_{\text{c}}}\right)$. The factor  $2N_{\text{r}}-N_{\text{t}}-1$ is contributed by   spatial diversity gain of using a large receive array with $N_r$ elements. 
\end{itemize}

\section{PC Estimation for Blind MIMO Receivers}\label{Sec:blind}
In the last section, an optimal coherent estimator was designed for PC estimation with receive CSI. The requirements on receive CSI and data statistics, say data SNR, are relaxed in this section so as to obviate the need of channel estimation and feedback, yielding the current case of blind PC estimation. 

\subsection{Integrated Design of Blind PC Estimator}

\subsubsection{Approximating the ML Problem} The lack of receive CSI complicates the ML problem in~\eqref{problem:ML}. It  is made tractable by finding a tractable lower bound on the likelihood function, which  is the objective, as follows. To this end, as the unknown channel is a Gaussian random matrix, the distribution of a received matrix symbol, say $\hat{\mathbf{Y}}_k$, conditioned on the ground truth $\mathbf{U}$ can be written as 
\begin{align}
	p(\hat{\mathbf{Y}}_k|\mathbf{U}) 
	&= \mathsf{E}_{\mathbf{Z}_k}[p(\hat{\mathbf{Y}}_k|\mathbf{U},\mathbf{Z}_k)],\nonumber\\
	&= \mathsf{E}_{\mathbf{Z}_k}\left[\frac{\exp\left(-\frac{1}{2}\mathsf{Tr}\left(\hat{\mathbf{Y}}_k(\mathbf{\Sigma}_k^{\prime})^{-1}\hat{\mathbf{Y}}_k^{\top}\right)\right)}{(2\pi)^{N_{\text{r}}N}\mathsf{det}(\mathbf{\Sigma}_k^{\prime})^{N_{\text{r}}}}\right] \label{Eq:CondDist},
\end{align}
where we recall $\mathbf{Z}_k$ to be the data noise and  define $\mathbf{\Sigma}_k^{\prime} = \sigma_{\text{c}}^2\mathbf{I}_N+(\mathbf{U} + \mathbf{Z}_k)^{\top}(\mathbf{U} + \mathbf{Z}_k)$. 

Next, as received symbols  are mutually independent, the likelihood function is given as 
\begin{align}
	\mathcal{L}(\mathbf{U;\mathbf{Y}}) 
	&= \ln\left(\prod_{k=1}^K p(\hat{\mathbf{Y}}_k|\mathbf{U})\right),\nonumber\\ 
	&= \sum_{k=1}^K\ln\left(\mathsf{E}_{\mathbf{Z}_k}\left[ p(\hat{\mathbf{Y}}_k|\mathbf{U}, \mathbf{Z}_k)\right]\right).\label{Eq:Objectve:blind}
\end{align}
It can be lower bounded using the Jensen's inequality as $\mathcal{L}(\mathbf{U;\mathbf{Y}})\geq \mathcal{L}_{\text{lb}}(\mathbf{U;\mathbf{Y}})$ with 
\begin{align}
	&\mathcal{L}_{\text{lb}}(\mathbf{U;\mathbf{Y}}) \nonumber\\
	&= \sum_{k=1}^K\mathsf{E}_{\mathbf{Z}_k}\left[\ln\left( p(\hat{\mathbf{Y}}_k|\mathbf{U}, \mathbf{Z}_k)\right)\right], \nonumber\\
	& = \sum_{k=1}^K\mathsf{E}_{\mathbf{Z}_k}\left[-\frac{1}{2}\mathsf{Tr}\left(\hat{\mathbf{Y}}_k(\mathbf{\Sigma}_k^{\prime})^{-1}\hat{\mathbf{Y}}_k^{\top}\right) - N_{\text{r}}\ln\left(\mathsf{det}(\mathbf{\Sigma}_k^{\prime})\right) \right] \nonumber\\
	&\,\quad- N_{\text{r}}NK\ln\left(2\pi\right). \label{Eq:Objectve:LB} 
\end{align}
For tractability, we approximate the objective of the ML problem in~\eqref{Eq:Objectve:blind} by its lower bound in~\eqref{Eq:Objectve:LB}. The new objective can be further simplified using the following result. 
\begin{Lemma}[] \label{Lemma:likelihoodlowerbound}
	\emph{
	Only the first term of  $\mathcal{L}_{\text{lb}}(\mathbf{U;\mathbf{Y}})$ in~\eqref{Eq:Objectve:LB}, i.e., $\sum_{k=1}^K\mathsf{E}_{\mathbf{Z}_k}\left[-\frac{1}{2}\mathsf{Tr}\left(\hat{\mathbf{Y}}_k(\mathbf{\Sigma}_k^{\prime})^{-1}\hat{\mathbf{Y}}_k^{\top}\right) \right]$, depends on the ground truth $\mathbf{U}$.}
\end{Lemma}
\begin{proof}
See Appendix~\ref{proof:likelihoodlowerbound}.
\end{proof}
Using the above result and changing the sign of its objective, the ML problem in~\eqref{problem:ML} can be approximated as the following minimization  problem 
\begin{equation}
    \begin{aligned}\label{eq:likelihood}
    	\mathop{\min}_{\mathbf{U}}&\ \ \frac{1}{K}\sum_{k=1}^K\mathsf{E}_{\mathbf{Z}_k}\left[\mathsf{Tr}\left(\hat{\mathbf{Y}}_k(\mathbf{\Sigma}_k^{\prime})^{-1}\hat{\mathbf{Y}}_k^{\top}\right)\right], \\
    	\mathrm{s.t.}&\ \ \mathbf{U}\mathbf{U}^{\top}=\mathbf{I}_M.
    \end{aligned}
\end{equation}

The above ML problem is not yet tractable and requires further approximation of its objective. For ease of exposition, represent  its summation term as $\Psi_k=\mathsf{E}_{\mathbf{Z}_k}\left[\mathsf{Tr}(\hat{\mathbf{Y}}_k(\mathbf{\Sigma}_k^{\prime})^{-1}\hat{\mathbf{Y}}_k^{\top})\right]$. Some useful notations are introduced as follows. Define a Gaussian matrix $\hat{\mathbf{Z}}_k =[\mathbf{I}_M\;\mathbf{0}_{M,N-M}] + \mathbf{Z}_k\mathbf{Q}^{\top}$ that  is a function of data noise (see Appendix~\ref{proof:ul_bounds}) and independent over $k$. Let $\hat{\mathbf{P}}_k = \mathcal{S}_M(\hat{\mathbf{Z}}_k^{\top}\hat{\mathbf{Z}}_k)^{\top}$, which  denotes  an orthonormal matrix representing the principal $M$-dimensional column subspace of  $\hat{\mathbf{Z}}_k$, and $\hat{\mathbf{S}}_k$ denote its singular value matrix. 
\begin{Lemma}\label{Lemma:ul_bounds}
\emph{The term $\Psi_k$ defined earlier can be rewritten as
\begin{equation}
	\Psi_k = C - \sigma_{\text{c}}^{-2}\mathsf{E}_{\hat{\mathbf{P}}_k,\hat{\mathbf{S}}_k}\left[\mathsf{Tr}\left(\hat{\mathbf{Y}}_k\mathbf{Q}^{\top}\hat{\mathbf{P}}_k\mathbf{\Lambda}_{\hat{\mathbf{S}}_k}\hat{\mathbf{P}}_k^{\top} \mathbf{Q}\hat{\mathbf{Y}}_k^{\top}\right)\right],
\end{equation}
where  the constant $C=\sigma_{\text{c}}^{-2}\mathsf{Tr}\left(\hat{\mathbf{Y}}_k\hat{\mathbf{Y}}_k^{\top}\right)$, the diagonal matrix $\mathbf{\Lambda_{\hat{\mathbf{S}}_k}} = \hat{\mathbf{S}}_k^2\left(\sigma_{\text{c}}^{-2}\mathbf{I}_M + \hat{\mathbf{S}}_k^2\right)^{-1}$, and $\mathbf{Q} = \left[\mathbf{U}^{\top}\;\left(\mathbf{U}^{\bot}\right)^{\top}\right]^{\top}$.}
\end{Lemma}
\begin{proof}
See Appendix~\ref{proof:ul_bounds}.
\end{proof}
Note that  $C$ is independent of the ground truth $\mathbf{U}$.
Therefore, it follows from Lemma~\ref{Lemma:ul_bounds} that  the maximization of $\Psi_k$ can be approximated by minimizing $\mathsf{E}_{\hat{\mathbf{P}}_k,\hat{\mathbf{S}}_k}\left[\mathsf{Tr}\left(\hat{\mathbf{Y}}_k\mathbf{Q}^{\top}\hat{\mathbf{P}}_k\mathbf{\Lambda}_{\hat{\mathbf{S}}_k}\hat{\mathbf{P}}_k^{\top} \mathbf{Q}\hat{\mathbf{Y}}_k^{\top}\right)\right]$. In other words, the current ML problem is equivalent to: 
\begin{equation}
    \begin{aligned} \label{eq:generalized PCA}
    	\mathop{\max}_{\mathbf{U}}&\ \ \frac{1}{K}\sum_{k=1}^K\mathsf{E}_{\hat{\mathbf{P}}_k,\hat{\mathbf{S}}_k}\left[\mathsf{Tr}\left(\hat{\mathbf{Y}}_k\mathbf{Q}^{\top}\hat{\mathbf{P}}_k\mathbf{\Lambda}_{\hat{\mathbf{S}}_k}\hat{\mathbf{P}}_k^{\top} \mathbf{Q}\hat{\mathbf{Y}}_k^{\top}\right)\right], \\ 
    	\mathrm{s.t.}&\ \ \mathbf{U}\mathbf{U}^{\top}=\mathbf{I}_M.
    \end{aligned}
\end{equation}
\subsubsection{Optimization  on Grassmannian} The above problem is solved using the theory of optimization on Grassmannian, involving the optimization of a subspace variable,  as follows. To this end, decompose the space of a received matrix  symbol, say $\hat{\mathbf{Y}}_k$, by eigen-decomposition as 
\begin{align}\label{eq:receivedsubspace}
	\hat{\mathbf{Y}}_k^{\top}\hat{\mathbf{Y}}_k & 
    = \mathbf{U}_{\hat{\mathbf{Y}}_k}^{\top}\mathbf{\Lambda}_{\textrm{dat}, k}\mathbf{U}_{\hat{\mathbf{Y}}_k} + \mathbf{V}_{\hat{\mathbf{Y}}_k}^{\top}\mathbf{\Lambda}_{\textrm{noi}, k}\mathbf{V}_{\hat{\mathbf{Y}}_k} \nonumber.
\end{align}
The eigenspace represented by  $\mathbf{U}_{\hat{\mathbf{Y}}_k}$ is the data space while that by $\mathbf{V}_{\hat{\mathbf{Y}}_k}$ is the channel-noise space. Here, each summation term $\mathsf{Tr}\left(\hat{\mathbf{Y}}_k\mathbf{Q}^{\top}\hat{\mathbf{P}}_k\mathbf{\Lambda}_{\hat{\mathbf{S}}_k}\hat{\mathbf{P}}_k^{\top} \mathbf{Q}\hat{\mathbf{Y}}_k^{\top}\right)$ directly depends on how close the PCs, say $\hat{\mathbf{P}}_k^{\top}\mathbf{Q}$, is to that of $\hat{\mathbf{Y}}_k^{\top}\hat{\mathbf{Y}}_k$, i.e. $\mathbf{U}_{\hat{\mathbf{Y}}_k}$. Specifically, we have
\begin{equation}
    \mathsf{Tr}\left(\hat{\mathbf{Y}}_k\mathbf{Q}^{\top}\hat{\mathbf{P}}_k\mathbf{\Lambda}_{\hat{\mathbf{S}}_k}\hat{\mathbf{P}}_k^{\top} \mathbf{Q}\hat{\mathbf{Y}}_k^{\top}\right) \leq \mathsf{Tr}\left(\mathbf{\Lambda}_{\textrm{noi}, k}\mathbf{\Lambda}_{\hat{\mathbf{S}}_k}\right)\nonumber,  
\end{equation}
where the upper bound is achieved if $d(\mathbf{Q}^{\top}\hat{\mathbf{P}}_k,\mathbf{U}_{\hat{\mathbf{Y}}_k}) =2M-2\mathsf{Tr}\left(\mathbf{U}_{\hat{\mathbf{Y}}_k}\mathbf{Q}^{\top}\hat{\mathbf{P}}_k\hat{\mathbf{P}}_k^{\top} \mathbf{Q}\mathbf{U}_{\hat{\mathbf{Y}}_k}^{\top}\right) = 0$. Therefore, we replace each trace term by the simplified term $\mathsf{Tr}\left(\mathbf{U}_{\hat{\mathbf{Y}}_k}\mathbf{Q}^{\top}\hat{\mathbf{P}}_k\hat{\mathbf{P}}_k^{\top} \mathbf{Q}\mathbf{U}_{\hat{\mathbf{Y}}_k}^{\top}\right)$ that focuses on projections onto a Grassmannian. It follows that the ML problem in~\eqref{eq:generalized PCA} can be approximated as 
\begin{equation}\label{pro:centroid}
    \begin{aligned}
	\mathop{\max}_{\mathbf{U}}&\ \ \frac{1}{K}\sum_{k=1}^K\mathsf{E}_{\hat{\mathbf{P}}_k}\left[\mathsf{Tr}\left(\mathbf{U}_{\hat{\mathbf{Y}}_k}\mathbf{Q}^{\top}\hat{\mathbf{P}}_k\hat{\mathbf{P}}_k^{\top} \mathbf{Q}\mathbf{U}_{\hat{\mathbf{Y}}_k}^{\top}\right)\right],\\
	\mathrm{s.t.}&\ \ \mathbf{U}\mathbf{U}^{\top}=\mathbf{I}_M.
	\end{aligned}
\end{equation}
Furthermore, the expectation $\mathsf{E}_{\hat{\mathbf{P}}_k}\left[\hat{\mathbf{P}}_k\hat{\mathbf{P}}_k^{\top}\right]$ shows the following property.
\begin{Lemma}\label{Lemma:mean:projection}
\emph{The orthogonal projection matrix $\hat{\mathbf{P}}_k\hat{\mathbf{P}}_k^{\top}$'s expectation is diagonal, say $\mathsf{E}_{\hat{\mathbf{P}}_k}\left[\hat{\mathbf{P}}_k\hat{\mathbf{P}}_k^{\top}\right]=\mathbf{\Lambda}_{\mu} = \mathsf{diag}\left(\mu_1,\mu_2,...,\mu_N\right)$ with $\mu_i \geq 0$, $i\in \{1,2,...,N\}$.} 
\end{Lemma}
\begin{proof}
See Appendix~\ref{proof:mean:projection}.
\end{proof}
Based on Lemma~\ref{Lemma:mean:projection}, define the summation $\mathbf{J}^{\prime} = \frac{1}{K}\sum_{k=1}^K\mathbf{U}_{\hat{\mathbf{Y}}_k}^{\top}\mathbf{U}_{\hat{\mathbf{Y}}_k}$ and the ML problem in~\eqref{pro:centroid} is written as 
\begin{equation}
    \begin{aligned} \label{eq:centroid}
    	\mathop{\max}_{\mathbf{U}}&\ \ \mathsf{Tr}\left(\mathbf{Q}^{\top}\mathbf{\Lambda}_{\mu}\mathbf{Q}\mathbf{J}^{\prime}\right),\\
    	\mathrm{s.t.}&\ \ \mathbf{U}\mathbf{U}^{\top}=\mathbf{I}_M.
    \end{aligned}
\end{equation}
Although the problem shown in \eqref{eq:centroid} is non-convex due to the orthonormal constraint, there still is a closed-form solution serving as the PC estimator for the blind case. To this end, define the SVD: $\mathbf{J}^{\prime} = \mathbf{U}_{\mathbf{J}^{\prime}}\mathbf{\Lambda}_{\mathbf{J}^{\prime}}\mathbf{U}_{\mathbf{J}^{\prime}}^{\top}$ with $\mathbf{U}_{\mathbf{J}^{\prime}}^{\top} = \mathcal{S}_N\left( \mathbf{J}^{\prime}\right)$. As the matrix $\mathbf{J}^{\prime}$ is positive semidefinite, the entries of $\mathbf{\Lambda}_{\mathbf{J}^{\prime}}$ are non-negative and we have the following  upper bound on  the objective function~\eqref{eq:centroid}: 
\begin{align}
	\mathsf{Tr}\left(\mathbf{Q}^{\top}\mathbf{\Lambda}_{\mu}\mathbf{Q}\mathbf{J}^{\prime}\right) 
	& = \mathsf{Tr}\left((\mathbf{Q}\mathbf{U}_{\mathbf{J}^{\prime}})^{\top}\mathbf{\Lambda}_{\mu}\mathbf{Q}\mathbf{U}_{\mathbf{J}^{\prime}}\mathbf{\Lambda}_{\mathbf{J}^{\prime}}\right),\nonumber\\
	& \leq \mathsf{Tr}\left(\mathbf{\Lambda}_{\mu}\mathbf{\Lambda}_{\mathbf{J}^{\prime}}\right)\nonumber.
\end{align}
The equality is achieved when  $\mathbf{Q} =\mathbf{U}_{\mathbf{J}^{\prime}}$. It follows that the $M$-dimensional principal eigenspace of $\mathbf{J}^{\prime}$,  $\mathcal{S}_M\left(\mathbf{J}^{\prime}\right)$, is the  optimal solution for the problem in \eqref{eq:centroid}. 
\begin{Theorem}[ML based PC estimation without receive CSI] \label{Theorem:blind}
	\emph{Given the received matrix symbols  $\{\hat{\mathbf{Y}}_k\}_{k=1}^K$ and the blind PC estimator based on ML criterion is given by
	\begin{equation}
	\tilde{\mathbf{U}}^{\star} = \mathcal{S}_M\left(\mathbf{J}^{\prime}\right),
	\end{equation}
	where $\mathbf{J}^{\prime} = \frac{1}{K}\sum_{k=1}^K\mathbf{U}_{\hat{\mathbf{Y}}_k}^{\top}\mathbf{U}_{\hat{\mathbf{Y}}_k}$ and $\mathbf{U}_{\hat{\mathbf{Y}}_k} = \mathcal{S}_M\left(\hat{\mathbf{Y}}_k^{\top}\hat{\mathbf{Y}}_k\right) $.}
\end{Theorem}

The optimal blind PC estimation in Theorem~\eqref{Theorem:blind} essentially  consists of three steps: 1) projecting each received matrix symbol to become a single point on the Grassmannian; 2) compute the centroid by averaging the points in the Euclidean space; 3) then projecting the result onto the Grassmannian to yield the estimated global PCs. It should be emphasized that Step 1) leverages the rotation-invariant property of analog subspace transmission  in \eqref{eq:rotation}. On the other hand, the  centroid computation in Step 2) is an aggregation operation that suppresses both the data and receiver noise. The aggregation gain grows as the number of devices or equivalently the number of observations grow as quantified in the next sub-section. 
	
\begin{Remark} [Geometric Interpretation] \label{Remark:Geometric}
	\emph{
	The result in Theorem~\ref{Theorem:blind} can be  also interpreted geometrically as follows. According to~\eqref{eq:distance}, each summation term  in the objective function in \eqref{eq:centroid}  measures the subspace  distance between a  received matrix symbol  and a Gaussian matrix, where the latter's PCs are uniformly displaced by $\mathbf{Q}$ on the Grassmannian. Since the Gaussian matrix is isotropic and its mean's PCs coincides with $\mathbf{U}$ after being displaced by $\mathbf{Q}$, we can obtain an equivalent geometric form of the problem in~\eqref{eq:centroid} as 
	\begin{equation}
    	\begin{aligned} 
        	\mathop{\min}_{\mathbf{U}}&\ \ \frac{1}{K}\sum_{k=1}^Kd^2(\mathbf{U},\mathbf{U}_{\hat{\mathbf{Y}}_k}),\nonumber\\
        	\mathrm{s.t.}&\ \ \mathbf{U}\mathbf{U}^{\top}=\mathbf{I}_M \nonumber.
        \end{aligned} 
	\end{equation}
    The above problem is to identify the optimal $\mathbf{U}$ such that the sum over its square distance to each signal's PCs is minimized. Its optimal solution, as proved in the literature is the centroid of signal PCs on the Grassmannian (see, e.g.~\cite{JF2019estimation2019}), which is aligned with the result in the theorem.}
\end{Remark}
The complexity of the blind estimator is analyzed as follows. The eigenvalue decomposition of the $M\times M$ matrix $\mathbf{J}^{\prime}$ incurs complexity in the order of $O(M^3)$. On the other hand, computing each summation term in $\mathbf{J}^{\prime}$, namely  $\mathbf{U}_{\hat{\mathbf{Y}}_k}^{\top}\mathbf{U}_{\hat{\mathbf{Y}}_k}$, has the complexity of  $O(M^2N)$. Given the received matrix $\hat{\mathbf{Y}}_k^{\top}\hat{\mathbf{Y}}_k$ with the size $2N_{\mathrm{r}}\times N$, the matrix multiplication of $\{\hat{\mathbf{Y}}_k^{\top}\hat{\mathbf{Y}}_k\}$ and their eigenvalue decomposition result in total complexity of $O(4KN_{\mathrm{r}}N^2)$. Therefore, the overall complexity for the blind estimator is $O(4KN_{\mathrm{r}}N^2+M^3 + KNM^2)$.

Finally, the blind ML PC estimator in Theorem~\ref{Theorem:blind} yields an unbiased estimate of the ground-truth as proved below.
\begin{Corollary}\label{Pro:unbiasedEwoCSI}
	\emph{The blind PC estimator in Theorem~\ref{Theorem:blind} computing the centroid of received noisy observations achieves unbiased estimation of the global PCs in the following sense. Define  $\bar{\mathbf{J}}^{\prime} = \mathsf{E}\left[\mathbf{J}^{\prime}\right]$ with $\mathbf{J}^{\prime}$ follows that in Theorem~\ref{Theorem:blind}, and the optimal PC estimate $\bar{\mathbf{U}}^\star = \mathcal{S}_M\left( \bar{\mathbf{J}}^{\prime}\right)$. Then $d(\bar{\mathbf{U}}^\star,\mathbf{U}) = 0$. }
\end{Corollary}
\begin{proof}
See Appendix~\ref{proof:unbiasedEwoCSI}.
\end{proof}
\subsection{Optimality of Symbol-by-Symbol Blind PC Detection} \label{Sec:separation}
The blind estimator designed in the preceding sub-section  performs joint detection of  a block of received symbols. However, its  aggregation form in Theorem~\ref{Theorem:blind}, which results from i.i.d. channel and Gaussian noise over users,  suggests the optimality  of detecting the subspaces of individual symbols followed by applying the conventional aggregation method for DPCA \cite{JF2019estimation2019,VC2020Arxiv}. The conjecture is corroborated in the sequel by designing the single-symbol blind PC detector under two different criteria. 

First, consider the ML criterion. Conditioned on the transmitted  symbol $\mathbf{U}_k$ and similarly as  \eqref{Eq:CondDist}, the distribution of $\hat{\mathbf{Y}}_k$ is obtained as 
\begin{align} 
    &p(\hat{\mathbf{Y}}_k|\mathbf{U}_k)\nonumber\\
    &=\frac{\exp\left(-\frac{1}{2}\mathsf{Tr}\left(\hat{\mathbf{Y}}_k^{\top}\left(\sigma_{\text{c}}^2\mathbf{I}_{N} +\mathbf{U}_k^{\top}\mathbf{U}_k\right)^{-1}\hat{\mathbf{Y}}_k\right)\right)}{(2\pi)^{N_{\mathrm{r}}N}\mathsf{det}\left(\sigma_{\text{c}}^2\mathbf{I}_{N} +\mathbf{U}_k^{\top}\mathbf{U}_k\right)^{N_{\mathrm{r}}}},\nonumber\\
    & =\frac{\exp\left(-\frac{1}{2\sigma_{\text{c}}^{2}}\mathsf{Tr}\left(\hat{\mathbf{Y}}_k^{\top}\hat{\mathbf{Y}}_k- (1+\sigma_{\text{c}}^{2})^{-1} \hat{\mathbf{Y}}_k^{\top}\mathbf{U}_k^{\top}\mathbf{U}_k\hat{\mathbf{Y}}_k\right)\right)}{(2\pi\sigma_{\text{c}}^{2})^{N_{\mathrm{r}}N}\left(1+\sigma_{\text{c}}^{-2} \right)^{N_{\mathrm{r}}M}} \nonumber.
 \end{align}
The above result shows that the distribution depends on $\mathbf{U}_k$ through the term $\mathsf{Tr}\left( \hat{\mathbf{Y}}_k^{\top}\mathbf{U}_k^{\top}\mathbf{U}_k\hat{\mathbf{Y}}_k\right)$. Therefore, the optimal ML estimate of $\mathbf{U}_k$ can be obtained by solving the following problem:
\begin{equation}\label{eq:sepration:commun}
	\begin{aligned} 
    	\mathop{\min}_{\mathbf{U}_k}&\ \ \mathsf{Tr}\left( \hat{\mathbf{Y}}_k^{\top}\mathbf{U}_k^{\top}\mathbf{U}_k\hat{\mathbf{Y}}_k\right),\\
    	\mathrm{s.t.}&\ \ \mathbf{U}_k\mathbf{U}_k^{\top}=\mathbf{I}_M.
    \end{aligned} 
\end{equation}
Note that even though $\mathbf{U}_k$ is not orthonormal, the orthonormal constraint in the above problem serves the purpose of matched filtering by reducing the noise dimensionality. Define  $\mathbf{U}_{\hat{\mathbf{Y}}_k}=\mathcal{S}_M(\hat{\mathbf{Y}}_k^{\top}\hat{\mathbf{Y}}_k)$.  The problem in  \eqref{eq:sepration:commun} has the optimal solution:  $\mathbf{U}_k^{\star} =\mathbf{Q}_{M}\mathbf{U}_{\hat{\mathbf{Y}}_k}$ with $\mathbf{Q}_{M}$ is an arbitrary $M\times M$ orthonormal matrix. Equivalently, $\mathsf{span}(\mathbf{U}_k^{\star})=\mathsf{span}(\mathbf{U}_{\hat{\mathbf{Y}}_k})$.

Second, in the absence of  receiver noise, the detection  problem can be formulated as linear regression: $\hat{\mathbf{Y}}_k=\tilde{\mathbf{H}}_k\tilde{\mathbf{U}}_k$ where  $\tilde{\mathbf{H}}_k$ and $\tilde{\mathbf{U}}_k$ denote the estimates of channel and transmitted subspace, respectively. The resultant regression error can be measured as  $\Vert\hat{\mathbf{Y}}_k-\tilde{\mathbf{H}}_k\tilde{\mathbf{U}}_k\Vert^2$. Then the problem of  minimizing the regression error can be solved by two steps: 1) given any $\tilde{\mathbf{U}}_k$, the optimal $\tilde{\mathbf{H}}_k$ is given by $\tilde{\mathbf{H}}_k^{\star}=\arg\min_{\tilde{\mathbf{H}}_k}\Vert\hat{\mathbf{Y}}_k-\tilde{\mathbf{H}}_k\tilde{\mathbf{U}}_k\Vert^2=\hat{\mathbf{Y}}_k\tilde{\mathbf{U}}_k^{\top}$; 2) given the optimal $\tilde{\mathbf{H}}_k^{\star}$, the optimal $\tilde{\mathbf{U}}_k$ is given by $\tilde{\mathbf{U}}_k^{\star}=\arg\min_{\tilde{\mathbf{U}}_k\tilde{\mathbf{U}}_k^{\top}=\mathbf{I}_M}\Vert\hat{\mathbf{Y}}_k-\hat{\mathbf{Y}}_k\tilde{\mathbf{U}}_k^{\top}\tilde{\mathbf{U}}_k\Vert^2=\mathbf{Q}_{M}\mathbf{U}_{\hat{\mathbf{Y}}_k}$. One can observe the result to be identical to the preceding  ML counterpart.

Finally, by setting the arbitrary orthonormal matrix $\mathbf{Q}_{M}$ to be an identity matrix, the aggregation of  subspaces estimated from individual received symbols  using the above single-symbol detector yields an identical result as in Theorem~\ref{Theorem:blind}. 

\subsection{Error Performance Analysis}
As before, the error performance of  blind PC estimation is also  analyzed based on the result in Lemma~\ref{Lemma:subspaceDistanceBetweenTwoPerturbedSubspace}. As the $\mathbf{U}_{\hat{\mathbf{Y}}_k}$ in the derived  estimator (see Theorem~\ref{Theorem:blind}) does not directly admit tractability, an approximation is necessary. Consider a  symbol resulting from applying channel inversion to a received symbol, say $\hat{\mathbf{Y}}_k$, resulting in 
 $\mathbf{U} + \mathbf{Z}_k + \hat{\mathbf{H}}_k^+\hat{\mathbf{W}}_k$. Note that this is a deviated version of the desired PCs  $\mathbf{U}_{\hat{\mathbf{Y}}_k}$ but its eigenspace can provide a tractable approximation  of the latter when the channel SNR is high~\cite{Yuqing2018AutomaticRecognition}: 
\begin{equation} \label{eq:approximation}
	\mathbf{U}_{\hat{\mathbf{Y}}_k}\approx \mathcal{S}_M\left((\mathbf{U} + \mathbf{Z}_k + \hat{\mathbf{H}}_k^+\hat{\mathbf{W}}_k)^{\top}(\mathbf{U} + \mathbf{Z}_k + \hat{\mathbf{H}}_k^+\hat{\mathbf{W}}_k)\right).
\end{equation}
Simulation results in Fig.~\ref{fig:distance} validate the approximation. Based on the above approximation, we can first derive the square subspace distance conditioned on the channel matrices $\{\hat{\mathbf{H}}_k\}_{k=1}^K$ as follows.
\begin{Lemma}[]\label{Lemma:EwoCSIdistance}\emph{
The square subspace distance between the estimation $\tilde{\mathbf{U}}^{\star}$ in the Theorem~\ref{Theorem:blind} and the ground-truth $\mathbf{U}$ is upper bounded as
\begin{align}
    d^2(\tilde{\mathbf{U}}^{\star},\mathbf{U})
    &\leq - \frac{2}{M^2}\sum_{k,j=1,k\neq j}^K\mathsf{Tr}\left( \mathbf{\Delta}_{2,k,j}\mathbf{U^{\bot}}^{\top}\mathbf{U}^{\bot}\mathbf{\Delta}_{2,k,j}^{\top}\right)  \nonumber\\
    &\,\quad +\frac{4}{M}\sum_{k=1}^K\mathsf{Tr}\left( \mathbf{\Delta}_{1,k}\mathbf{U^{\bot}}^{\top}\mathbf{U}^{\bot}\mathbf{\Delta}_{1,k}^{\top}\right) + O\left(\epsilon^3\right),
\end{align}
where $\mathbf{U^{\bot}} $ is the orthogonal complement of $ \mathbf{U}$,  $\epsilon = \sigma_{\text{d}}^2 + \frac{\sigma_{\text{c}}^2}{2N_{\text{r}}-M-1}$, $\mathbf{\Delta}_{1,k} \sim\mathcal{MN}[\mathbf{0},\sigma_{\text{d}}^2\mathbf{I}_M+\sigma_{\text{c}}^2\hat{\mathbf{H}}_k^+(\hat{\mathbf{H}}_k^{+})^{\top} ,\mathbf{I}_N]$, and
\begin{align*}
    &\mathbf{\Delta}_{2,k,j}\sim \\
    &\quad\mathcal{MN}[\mathbf{0},2\sigma_{\text{d}}^2\mathbf{I}_M+\sigma_{\text{c}}^2\hat{\mathbf{H}}_k^+(\hat{\mathbf{H}}_k^{+})^{\top} + \sigma_{\text{c}}^2\hat{\mathbf{H}}_j^+(\hat{\mathbf{H}}_j^{+})^{\top},\mathbf{I}_N].
\end{align*}}

Proof: \emph{See Appendix F.}
\end{Lemma}
Similarly, when both data SNR and channel SNR are large, i.e. $\epsilon\rightarrow 0$, the residual term converges to zero. Then, taking an expectation over channel matrices and the upper bound of MSSD in the case of blind MIMO receivers can be obtained.
\begin{figure}
	\centering
	\includegraphics[scale=0.57]{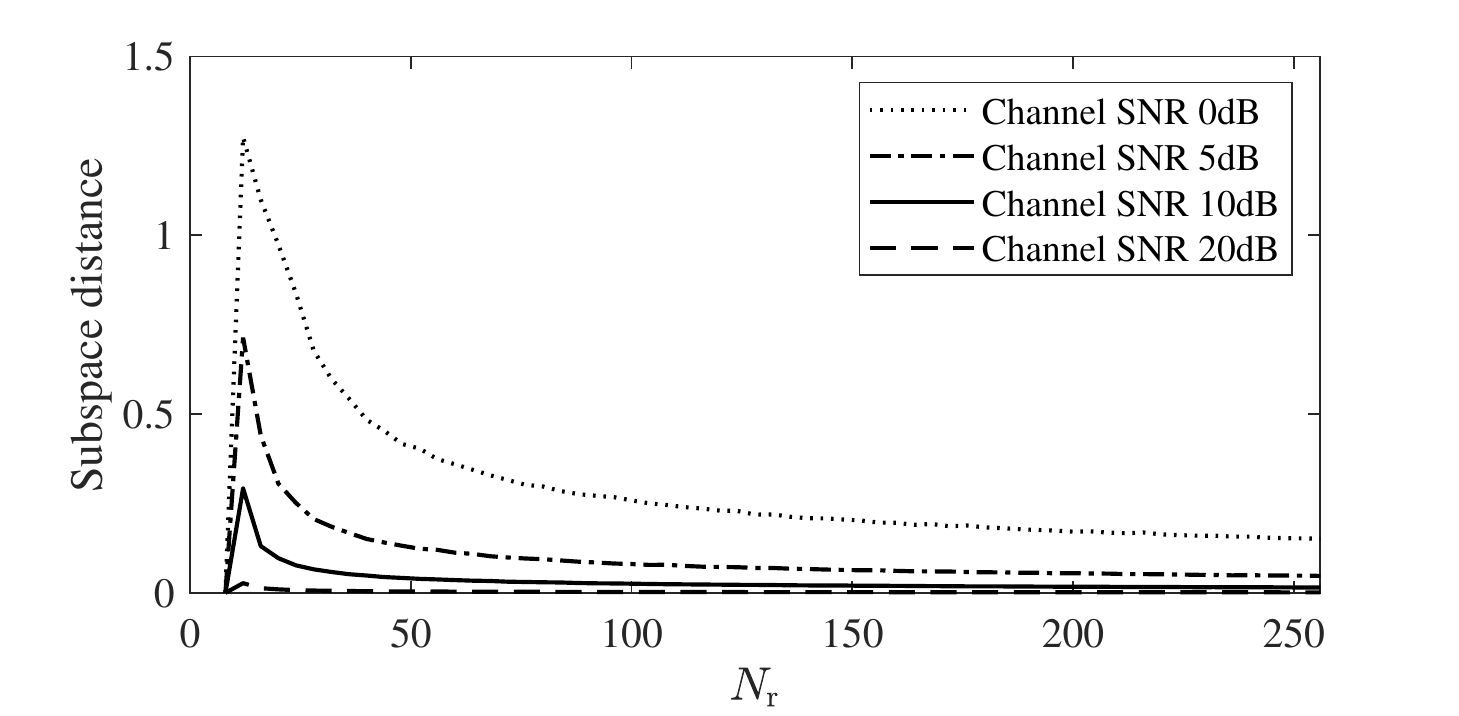}
	\caption{{Subspace distance between the spaces spanned by PCs $\mathbf{U}_{\hat{\mathbf{Y}}_k}$ and $\mathcal{S}_M\left((\mathbf{U} + \mathbf{Z}_k + \hat{\mathbf{H}}_k^+\hat{\mathbf{W}}_k)^{\top}(\mathbf{U} + \mathbf{Z}_k + \hat{\mathbf{H}}_k^+\hat{\mathbf{W}}_k)\right)$ with  $N_{\text{t}} = M = 8$.} }
	\label{fig:distance}
\end{figure}

\begin{Theorem}[Error of Proposed Blind PC Estimation]\label{The:upperBoundMSSDEwoCSI}\emph{
Given many devices ($K\rightarrow \infty$) and small channel and data noise ($\sigma_{\text{d}},\sigma_{\text{c}}\rightarrow 0$), the MSSD of the blind PC estimator in Theorem~\ref{Theorem:blind} can be asymptotically upper bounded as
\begin{equation}
    d_{\text{ms}}(\tilde{\mathbf{U}}^{\star},\mathbf{U})\leq \frac{4M(N-M)}{K}\left(\sigma_{\text{d}}^2 + \frac{\sigma_{\text{c}}^2}{2N_{\text{r}}-M-1}\right).
\end{equation}
for a MIMO system (i.e., $2N_r -M > 1$). }
\end{Theorem}
\begin{proof}
See Appendix~\ref{proof:upperBoundMSSDEwoCSI}.
\end{proof}
Using the definition of $\sigma_{\text{d}}$ and $\sigma_{\text{c}}$ yields the following result:
\begin{equation}
d_{\text{ms}}(\tilde{\mathbf{U}}^{\star},\mathbf{U}) \leq \frac{4M^2(N-M)}{KN}\left[\gamma_{\text{d}}^{-1} + \frac{(1+M\gamma_{\text{d}}^{-1})\gamma_{\text{c}}^{-1}}{2N_{\text{r}}-M-1}\right]\nonumber.
\end{equation}
\begin{Remark} [Coherent vs. Blind] \label{Remark:comparison}
\emph{Comparing Theorems~\ref{The:upperBoundMSSDEwCSI} and~\ref{The:upperBoundMSSDEwoCSI}, the error bounds for both the coherent and blind PC estimators are observed to have the same form except for the difference by  a multiplicative factor of two. This shows that despite the lack of receive CSI, the latter can achieve similar performance as its coherent counterpart mainly due to the exploitation of the channel-rotation-invariance of analog  subspace transmission in~\eqref{eq:model}. On the other hand, the said multiplicative factor represents the additional gain in estimation error suppression 
via regularized  channel inversion that alleviates the effects of  data and receiver noise.}
\end{Remark}

\section{Extension to PCs with Arbitrary Dimensions} \label{Sec:extension}
In the previous sections, we assume  $N_{\text{t}} = M$ such that  each matrix symbol representing  local PCs can be directly transmitted over a transmit array. The results can be generalized by relaxing   the assumption as follows. 

First, consider the case of $N_{\text{t}} \geq M$. One  can introduce an $N_{\text{t}}$-by-$M$ orthonormal matrix $\mathbf{X}$ that maps an $M$-dimensional symbol to the $N_t$-element array, yielding the following communication model: 
\begin{align}
	\mathbf{Y}_k
	& = \mathbf{H}_k\mathbf{X}(\mathbf{U} + \mathbf{Z}_k) + \mathbf{W}_k \nonumber\\
	& \overset{\triangle}{=}\tilde{\mathbf{H}}_k(\mathbf{U} + \mathbf{Z}_k) + \mathbf{W}_k \nonumber,
\end{align}
where $\tilde{\mathbf{H}}_k \overset{\triangle}{=} \mathbf{H}_k\mathbf{X}$. With transmit CSI, $\mathbf{X}$ can be generated isotropically. As a result, $\tilde{\mathbf{H}}_k$  follows an isotropic matrix Gaussian $\mathcal{MN}(\mathbf{0},\mathbf{I}_{N_{\text{r}}},\mathbf{I}_M)$ and thus the results in the preceding sections remain unchanged.

Next, consider the case of $M > N_{\text{t}}$. Local PCs at a device have to be transmitted as multiple matrix symbols. Let a local estimate, say $\mathbf{U}_{k}$, be partitioned into $T$ $M' \times N$ component matrices with $M' \leq  N_{\text{t}}$ as follows
\begin{equation}
	\mathbf{U}_{k} = [\mathbf{U}_{k,1}^{\top},\mathbf{U}_{k,2}^{\top},...,\mathbf{U}_{k,T}^{\top}]^{\top} \nonumber.
\end{equation}
Since each component matrix contains $M'$ columns of the orthonormal matrix of $\mathbf{U}_{k}$, it is also orthonormal. By introducing a mapping matrix $\mathbf{X}'$ similarly as in the preceding case, each component matrix, say $\mathbf{U}_{k,t}$, can be transmitted as a single matrix symbol, resulting in $\mathbf{Y}_{k,t} = \mathbf{H}_{k,t}\mathbf{X}'\mathbf{U}_{k,t} + \mathbf{W}_{k,t}$. Obviously, if $M' = N_t$, the mapping $\mathbf{X}'$ is unnecessary and the above model reduces to $\mathbf{Y}_{k,t} = \mathbf{H}_{k,t}\mathbf{U}_{k,t} + \mathbf{W}_{k,t}$. At the receiver, multiple  received component matrices are combined into a single one in a similar way as in 
 \eqref{eq:model}: 
\begin{align}
	\mathbf{Y}_{k}
	& = [\mathbf{Y}_{k,1}^{\top},\mathbf{Y}_{k,2}^{\top},...,\mathbf{Y}_{k,T}^{\top}]^{\top} \nonumber \\
	& = \mathsf{diag}\left(\mathbf{H}_{k,1},...,\mathbf{H}_{k,T}\right)
	\mathbf{U}_{k} + [\mathbf{W}_{k,1}^{\top},...,\mathbf{W}_{k,T}^{\top}]^{\top} \label{eq:combination2}.
\end{align}
Therefore, we can obtain a communication model~\eqref{eq:combination2} differing from that  in~\eqref{eq:model} only in the MIMO channel which is block diagonal in the current case. For coherent PC estimation, it is straightforward to show that the design modified using the generalized channel model retains its unbiasedness and ML optimality. As for blind PC estimation,  $\mathsf{diag}(\mathbf{H}_{k,1},...,\mathbf{H}_{k,T})$ keeps the isotropic property as before that guarantees unbiased estimation. Furthermore, to modify the design and analysis,  the expectation $\mathsf{E}_{\mathbf{H}_k}\left[\hat{\mathbf{H}}_k^+(\hat{\mathbf{H}}_k^{+})^{\top}\right]$ is replaced with  
\begin{align*}
	&\mathsf{E}_{\mathbf{H}_k}\left[\hat{\mathbf{H}}_k^+(\hat{\mathbf{H}}_k^{+})^{\top}\right] \\
	&= \mathsf{diag}\left(\mathsf{E}_{\mathbf{H}_{k,1}}[\hat{\mathbf{H}}_{k,1}^{+}(\hat{\mathbf{H}}_{k,1}^{+})^{\top}],...,\mathsf{E}_{\mathbf{H}_{k,T}}[\hat{\mathbf{H}}_{k,T}^{+}(\hat{\mathbf{H}}_{k,T}^{+})^{\top}]\right),
\end{align*}
where $\hat{\mathbf{H}}_{k,t}$ is derived from $\mathbf{H}_{k,t}$ using  the method in \eqref{eq:newModel}. This modification, according to the derivation process, only changes the coefficients of $\sigma_{\text{c}}^2$ in both Theorem~\ref{The:upperBoundMSSDEwCSI} and~\ref{The:upperBoundMSSDEwoCSI}. Therefore, one can conclude that the scaling laws of estimation error are preserved during the extension.

\section{Simulation Results}\label{Sec:experiments}
The default simulation settings are as follows. The PC dimensions are taken to be $M\times N=8\times 200$, following common settings in the DPCA literature (see, e.g.~\cite{JF2019estimation2019,VC2020Arxiv}). The array sizes at the edge devices and the server are $N_{\text{t}} = 8$ and $N_{\text{r}} = 16$, respectively, which is also one configuration adopted in the 3GPP standard.
\begin{figure*}[t]
    \centering
    \subfigure[$\gamma_{\text{d}} = \gamma_{\text{c}}=5$dB]{\label{subfig:AnavsDig_5dB}\includegraphics[width=0.4\textwidth]{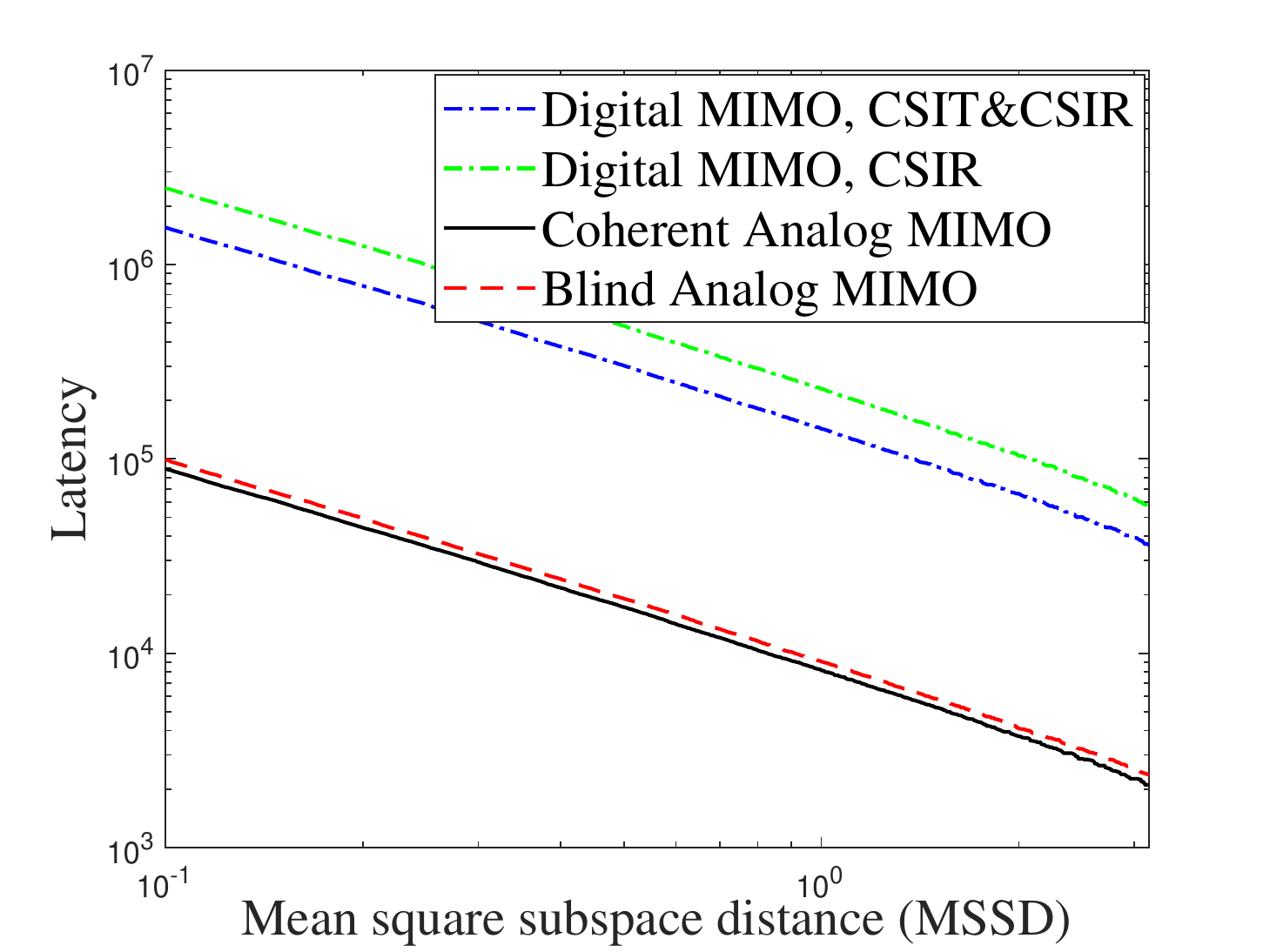}}
    \subfigure[$\gamma_{\text{d}} = \gamma_{\text{c}}=10$dB]{\label{subfig:AnavsDig_10dB}\includegraphics[width=0.4\textwidth]{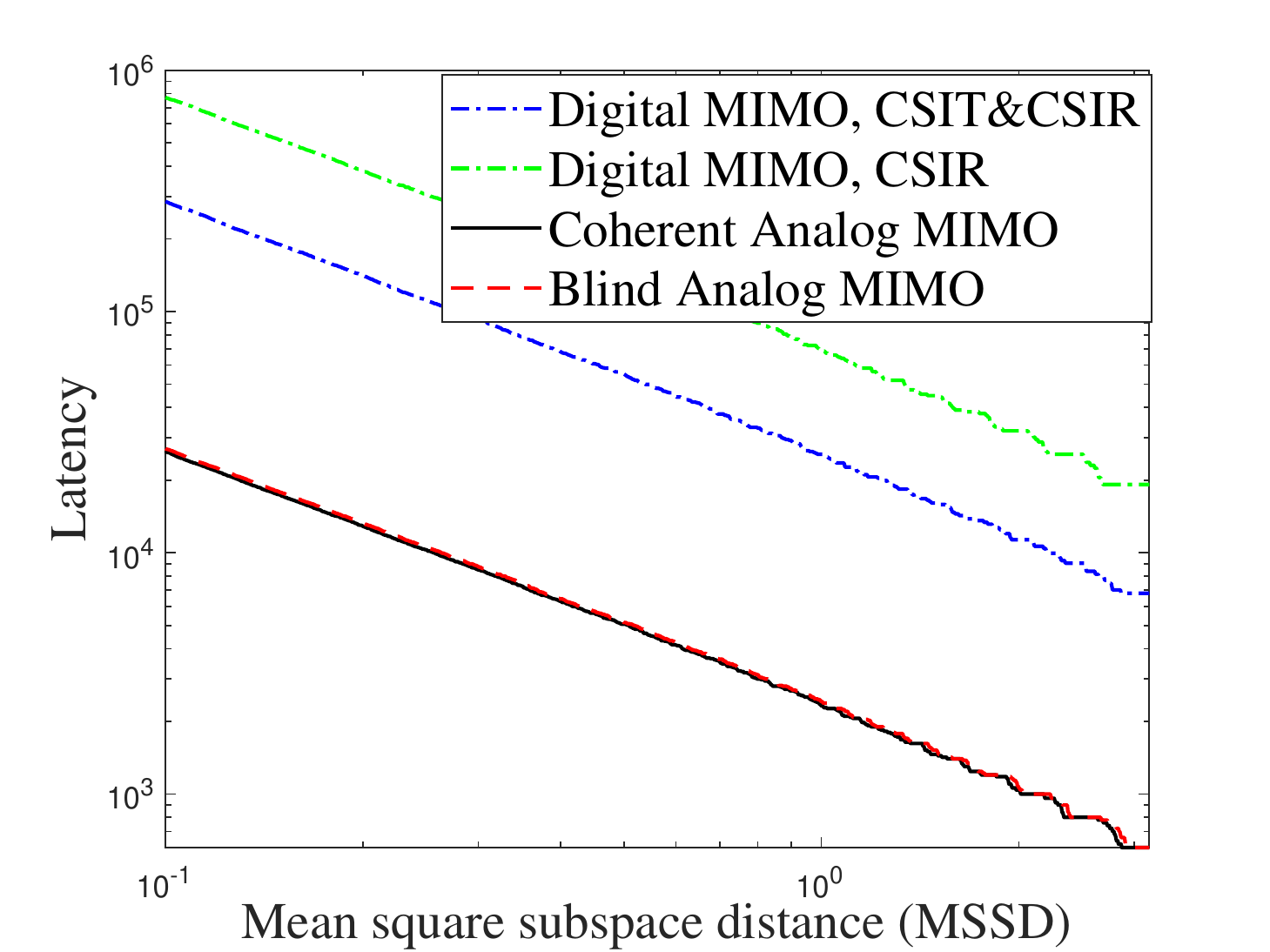}}
	\caption{Latency performance comparison between digital transmission and analog transmission.}
	\label{fig:digitalvsanalog}
\end{figure*}

\begin{figure*}[t]
    \centering
    \subfigure[Coherent estimation]{\label{subfig:coherent_device}\includegraphics[width=0.26\textwidth]{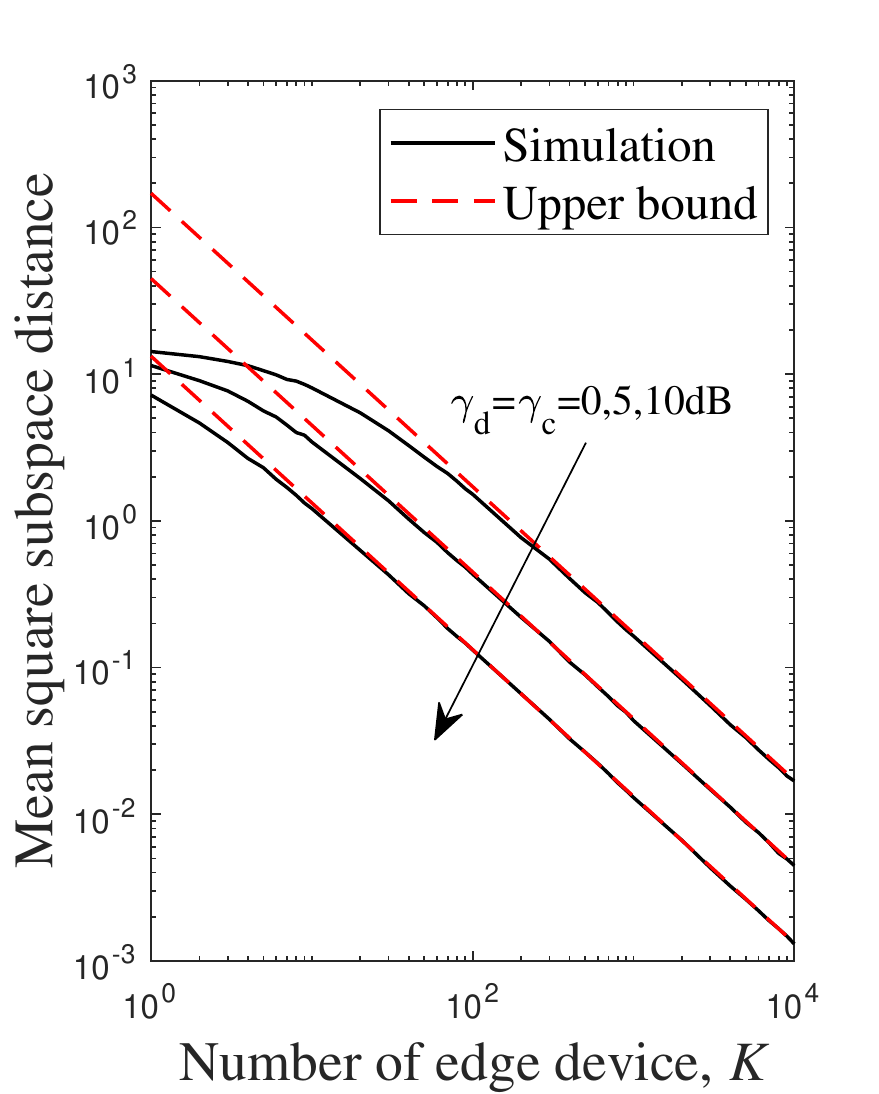}}
    \subfigure[Blind estimation]{\label{subfig:blind_device}\includegraphics[width=0.26\textwidth]{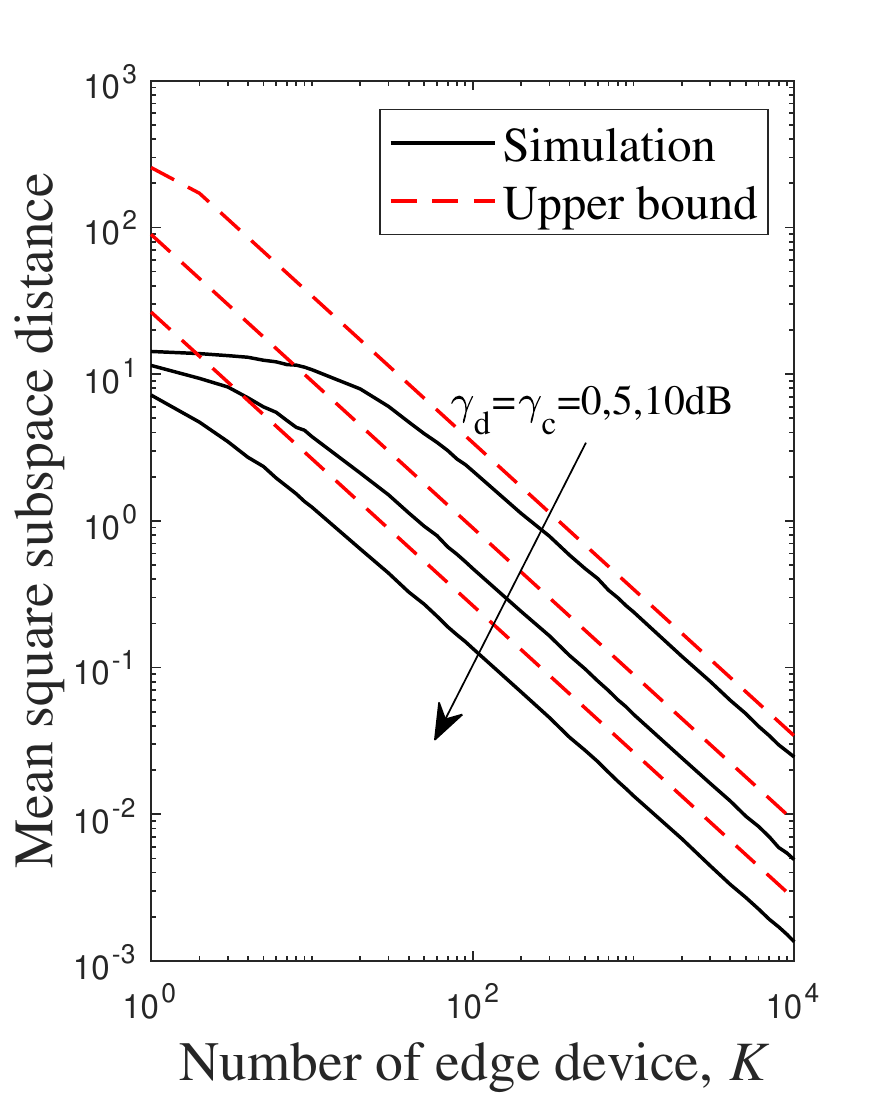}}
    \subfigure[Comparison]{\label{subfig:comparison_device}\includegraphics[width=0.26\textwidth]{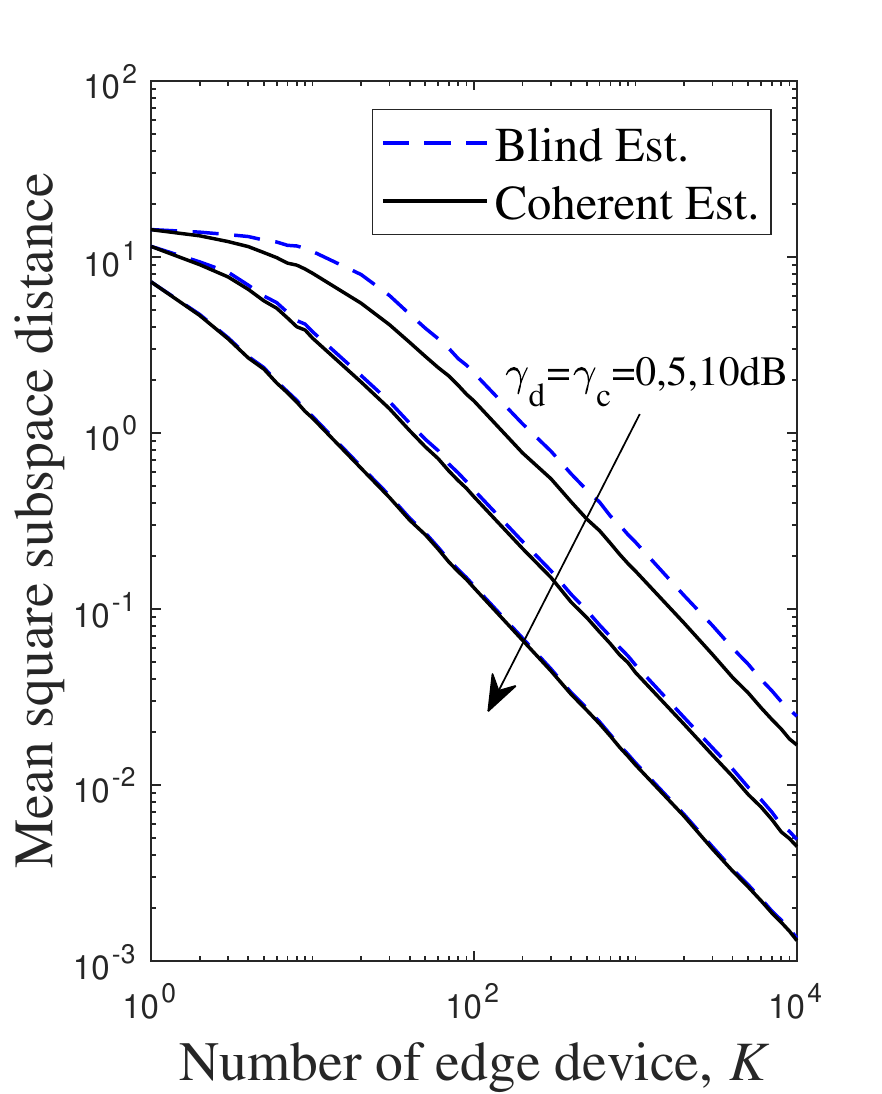}}
    \caption{PC estimation  performance comparisons for varying numbers of devices between  simulation and analysis for (a) coherent estimation and (b) blind estimation, and (c) comparison between the two designs.}
    \label{fig:edge_device_number}
\end{figure*}

\begin{figure*}[t]
    \centering
    \subfigure[Coherent estimation]{\label{subfig:coherent_SNR}\includegraphics[width=0.4\textwidth]{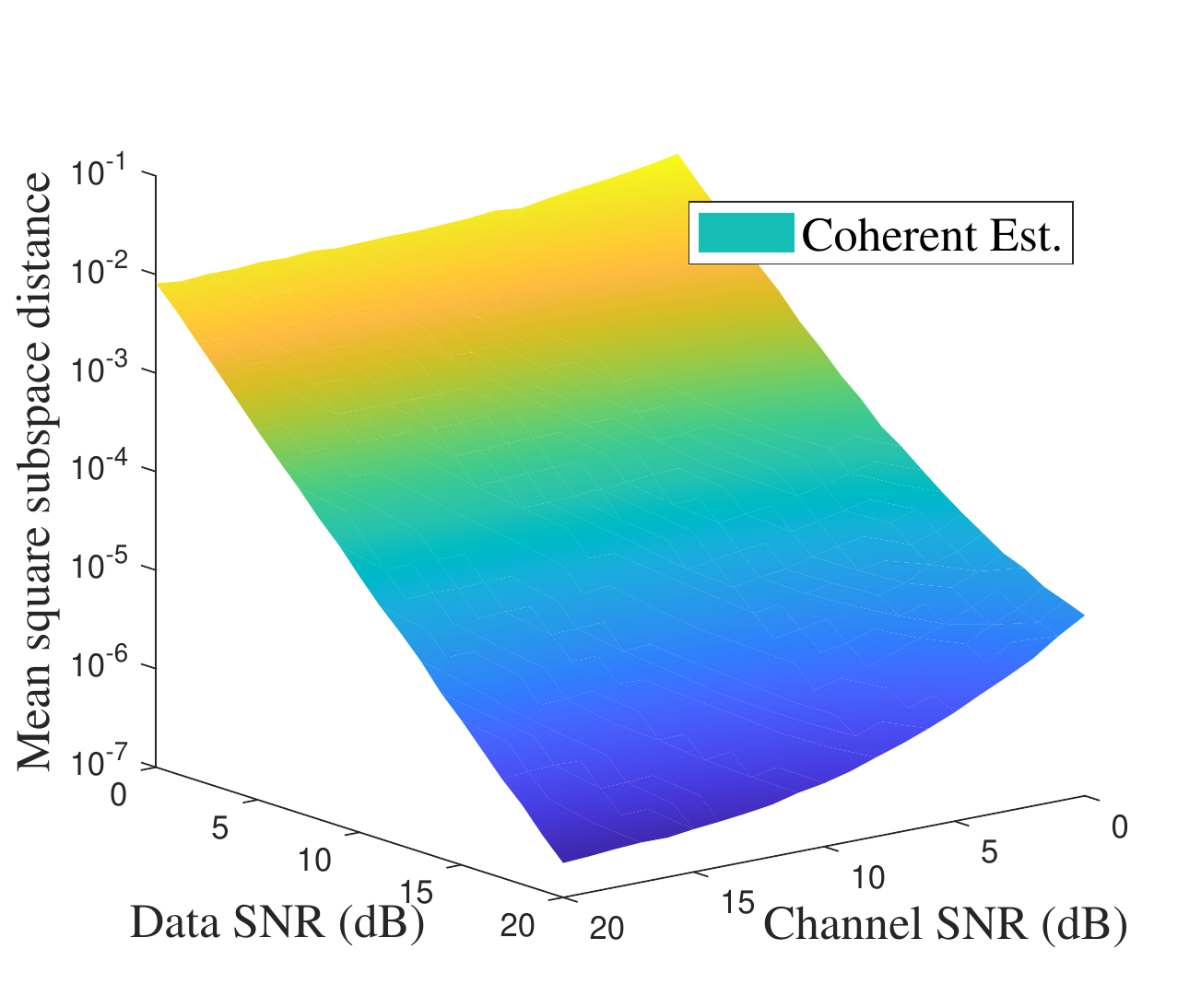}}
    \subfigure[Blind estimation]{\label{subfig:blind_SNR}\includegraphics[width=0.4\textwidth]{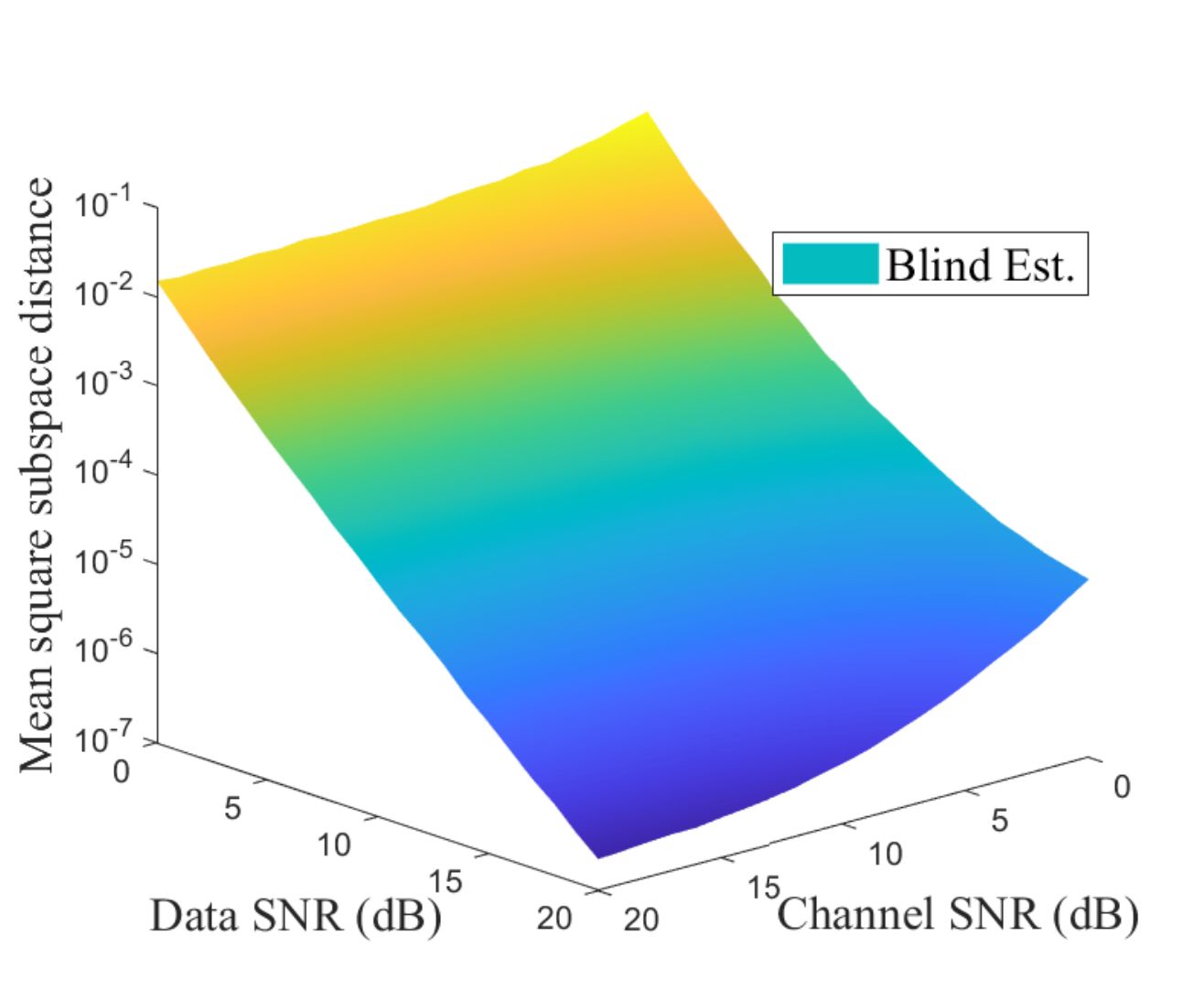}}
	\caption{Mean square subspace distance versus data and channel SNR.}
	\label{fig:MSSDvsSNR}
\end{figure*}
First, in Fig.~\ref{fig:digitalvsanalog}, we benchmark the latency performance of the proposed analog-MIMO enabled DPCA against the conventional designs using digital MIMO for varying target MSSD. The DPCA latency refers the number of channel uses (or symbol slots) needed to achieve the target MSSD. The benchmark digital-MIMO schemes  reliably upload local PCs to the server using $16$-bit quantization,  M-ary Quadrature Amplitude Modulation (M-QAM) modulation  and channel coding (i.e., convolution coding with rate 0.5~\cite{QingwenAMC2004}); the uploaded local PCs are aggregated to yield the global PCs using the DPCA scheme with uniform aggregation in \cite{JF2019estimation2019}. Two schemes are considered. The first scheme with only receive CSI uses  a linear ZF receiver,  fixes the M-QAM modulation and applies uniform transmission-power allocation due to the lack of transmit CSI~\cite{UniformPL2003}. The other scheme with both the transmit and receive CSI employs eigen-mode precoding and receiver  as well as water-filling power allocation and adaptive M-QAM modulation jointly optimized to maximize the spectrum efficiency~\cite{Telatar1999,Goldsmith2001}. It can be observed from Fig.~\ref{fig:digitalvsanalog} that although analog MIMO without coding is directly exposed to perturbation by receiver noise, its advantage of ultra-low latency greatly outweighs the disadvantage. Specifically, analog MIMO achieves more than a $10$-times latency reduction over its two digital counterparts  across  different SNRs. It is worth recalling  that the robustness of analog MIMO  is  due to the insensitivity  of the transmitted subspaces towards channel rotation and that the aggregation operation of DPCA suppresses  receiver noise. Moreover, the latency of different schemes are found to decrease linearly as the target MSSD grows while their latency ratios remain approximately constant over the MSSD range.  

Second, the effect of device population on the PC estimation is investigated. To this end, the curves of MSSD versus the number of devices $K$, are plotted in Fig.~\ref{fig:edge_device_number}, for which  the data and channel  SNRs are set as $\gamma_{\text{d}} = \gamma_{\text{c}} = \{0, 5, 10\}$ dB. The scaling laws derived in Theorems~\ref{Theorem:coherent} and~\ref{Theorem:blind} are validated using simulation results in Figs.~\ref{subfig:coherent_device} and~\ref{subfig:blind_device}, for the designs of coherent and blind PC estimation, respectively. The performance comparison between the two designs is provided in Fig.~\ref{subfig:comparison_device}. We observe that they follow the same linear scaling law for coherent and blind estimation, namely that MSSD decreases linearly as $K$ grows. Varying the SNRs shifts the curves but does not change their slopes. Next, one can observe from Figs.~\ref{subfig:coherent_device} and~\ref{subfig:blind_device} that the performance analysis accurately quantifies the linear scaling and its rate. The  derived MSSD bound is tight for the case of coherent estimation and a sufficiently large number of devices (e.g., $K\geq 10$) while for the other case of blind estimation the bound is not as tight but satisfactory. More importantly, the performance comparison in Fig.~\ref{subfig:comparison_device} confirms a main conclusion from the  analysis (see Remark~\ref{Remark:comparison}) that blind PC estimation  performs only with a small gap from that of the coherent counterpart especially at sufficiently high SNRs, advocating fast DPCA without CSI. 

Finally, the effects of data/channel SNRs on the performance of  analog MIMO are characterized in Fig.~\ref{fig:MSSDvsSNR} where the curves of  MSSD versus data/channel SNR are plotted for $K=1000$ devices. In the high channel-SNR regime, one can observe that increasing the data SNR reduces the MSSD of both the coherent and blind estimators at an approximately \emph{linear rate}. On the other hand, with the data SNR fixed, the MSSD is relatively insensitive to channel SNR especially in the low data-SNR regime. This is due to the fact that the diversity gain of multiple antenna systems suppressing receiver noise and large data noise becomes dominant, which is consistent with our analysis results shown in Theorem~\ref{The:upperBoundMSSDEwCSI} and~\ref{The:upperBoundMSSDEwoCSI}.

%
\section{Concluding Remarks}\label{Sec:conclusion}
We have presented an analog MIMO communication framework for supporting ultra-fast DPCA in a wireless system. Besides spatial multiplexing inherent in MIMO, the proposed approach dramatically reduces the DPCA latency by avoiding quantization, coding, and CSI feedback. Furthermore, by analyzing and comparing the performance of analog MIMO with and without receive CSI, we have found that even channel estimation may not be necessary as the resultant performance gain is small. This advocates blind analog MIMO to further accelerate DPCA. Compared with digital MIMO, more than a tenfold latency reduction has been demonstrated for analog MIMO in achieving the same  DPCA error performance. If aggregating data from many devices, the time sharing protocol may incur too high latency. Two solutions warranting further investigation are to use AirAggregate~\cite{KYang2020,GXZhuonebit2020,GXZhuAirComp2019} techniques and to equip the server with a large-scale array that can be used to resolve multiple spatially multiplexed transmissions from the devices~\cite{Erik2022SPAWC}. 

Building on findings in this work, we believe that analog communication will play an important role in machine learning and data analytics at the network edge. This requires new techniques to properly  control the effects of channel distortion and noise so that they can be masked by those of data noise. Moreover, the new techniques should be customized for the algorithms of specific applications (e.g., the current   principal-component estimation for DPCA) to optimize the system efficiency and performance. 
%

\appendix

\subsection{Proof of Lemma~\ref{Lemma:EwCSIdistance}}
\label{proof:EwCSIdistance}
The orthonormal $\tilde{\mathbf{U}}^{\star}$ given in Theorem 1 represents the principal eigenspace of $\mathbf{J}$, i.e. $\tilde{\mathbf{U}}^{\star} = \mathcal{S}_M\left({\mathbf{J}}^{\top}\mathbf{J}\right)$, where $\mathbf{J}=\frac{1}{K}\mathbf{\Sigma}_{\mathbf{E}}\mathbf{U} + \mathbf{E}$ with $\mathbf{\Sigma}_{\mathbf{E}} = \sum_{k=1}^{K} \left[\sigma_{\text{d}}^2\mathbf{I}_M+\sigma_{\text{c}}^2\hat{\mathbf{H}}_k^+(\hat{\mathbf{H}}_k^{+})^{\top}\right]^{-1}$.
Due to $\mathbf{Z}_k + \hat{\mathbf{H}}_k^+ \hat{\mathbf{W}}_k \sim \mathcal{MN}[\mathbf{0},\sigma_{\text{d}}^2\mathbf{I}_M+\sigma_{\text{c}}^2\hat{\mathbf{H}}_k^+(\hat{\mathbf{H}}_k^{+})^{\top},\mathbf{I}_N],\forall k$, $\mathbf{E}$ follows from a matrix Gaussian distribution $\mathcal{MN}[\mathbf{0}, \mathbf{\Sigma}_{\mathbf{E}},\mathbf{I}_N]$. Note that $\mathbf{\Sigma}_{\mathbf{E}}$ is a positive definite Hermitian matrix and there always exists an inverse matrix $\mathbf{\Sigma}_{\mathbf{E}}^{-1}$. Then, the rotation-invariant property of subspace in \eqref{eq:rotation} gives
\begin{align}
    \tilde{\mathbf{U}}^{\star} 
    & = \mathcal{S}_M\left((\mathbf{\Sigma}_{\mathbf{E}}\mathbf{U} + \mathbf{E})^{\top}(\mathbf{\Sigma}_{\mathbf{E}}\mathbf{U} + \mathbf{E})\right) \nonumber\\
    & = \mathcal{S}_M\left((\mathbf{U} + \mathbf{\Sigma}_{\mathbf{E}}^{-1}\mathbf{E})^{\top}(\mathbf{U} + \mathbf{\Sigma}_{\mathbf{E}}^{-1}\mathbf{E})\right)\nonumber,
\end{align}
where $\mathbf{\Sigma}_{\mathbf{E}}^{-1}\mathbf{E}$ is matrix Gaussian $\mathcal{MN}[\mathbf{0}, \mathbf{\Sigma}_{\mathbf{E}}^{-1},\mathbf{I}_N]$. Then, based on Lemma~\ref{Lemma:subspaceDistanceBetweenTwoPerturbedSubspace}, set $\epsilon = C_1\Vert \mathbf{\Sigma}_{\mathbf{E}}^{-1}\Vert_2$ with a constant $C_1$, $\mathbf{E}_1 = \frac{1}{\epsilon}\mathbf{\Sigma}_{\mathbf{E}}^{-1}\mathbf{E}$ and $\mathbf{E}_2 = \mathbf{O}$ and we can get the final result.

\subsection{Proof of Theorem~\ref{The:upperBoundMSSDEwCSI}}\label{proof:upperBoundMSSDEwCSI}
Based on Lemma~\ref{Lemma:EwCSIdistance}, with the second residual term omitted, the MSSD is approximated as
\begin{align}
	d_{\text{ms}}(\tilde{\mathbf{U}}^{\star},\mathbf{U})
	& \approx \mathsf{E}\left[2\mathsf{Tr}\left( \mathbf{\Delta}_{\mathbf{E}}\mathbf{U^{\bot}}^{\top}\mathbf{U}^{\bot}{\mathbf{\Delta}_{\mathbf{E}}}^{\top}\right)\right],\nonumber\\
	& = 2(N-M)\mathsf{E}_{\{\mathbf{H}_k\}_{k=1}^K}\left[\mathsf{Tr}\left(\mathbf{\Sigma}_{\mathbf{E}}^{-1}\right)\right], \nonumber\\
	& \leq 2(N-M)\mathsf{Tr}\left(\left(\sum_{k=1}^{K} \mathsf{E}_{\mathbf{H}_k}[\mathbf{A}_k]^{-1}\right)^{-1}\right)\nonumber,
\end{align}
where we define $\mathbf{A}_k=\sigma_{\text{d}}^2\mathbf{I}_M+\sigma_{\text{c}}^2\hat{\mathbf{H}}_k^+(\hat{\mathbf{H}}_k^{+})^{\top}$ and the last inequality comes from Jensen's inequality and the fact that $\mathsf{Tr}\left((\sum_{k=1}^K \mathbf{A}_k^{-1})^{-1} \right)$ is concave over $\{\mathbf{A}_k\}_{k=1}^K$~\cite{Concave2004}. Since $\hat{\mathbf{H}}_k^+(\hat{\mathbf{H}}_k^+)^{\top} = (\hat{\mathbf{H}}_k^{\top}\hat{\mathbf{H}}_k)^{-1}$ is an inverse Wishart matrix, given $2N_{\text{r}} - M >1$, we have the first moment $\mathsf{E}[(\hat{\mathbf{H}}_k^{\top}\hat{\mathbf{H}}_k)^{-1}] = \frac{1}{2N_{\text{r}}-M-1}\mathbf{I}_M$~\cite{Dietrich1988}. Using the above results yields the following upper bound:
\begin{align*}
	&d_{\text{ms}}(\tilde{\mathbf{U}}^{\star},\mathbf{U}) \\
	& \leq 2(N-M)\mathsf{Tr}\left(\left(\sum_{k=1}^{K} \left(\sigma_{\text{d}}^2+\frac{\sigma_{\text{c}}^2}{2N_{\text{r}}-M-1}\right)^{-1}\mathbf{I}_M\right)^{-1}\right),\\
	& = \frac{2M(N-M)}{K}\left(\sigma_{\text{d}}^2+\frac{\sigma_{\text{c}}^2}{2N_{\text{r}}-M-1}\right).
\end{align*} 

\subsection{Proof of Lemma~\ref{Lemma:likelihoodlowerbound}}\label{proof:likelihoodlowerbound}
In $\mathcal{L}_{\text{lb}}(\mathbf{U;\mathbf{Y}})$, the third term is constant. For the second part, we have 
\begin{align*}
	&\mathsf{det}\left(\sigma_{\text{c}}^2\mathbf{I}_N+(\mathbf{U} + \mathbf{Z}_k)^{\top}(\mathbf{U} + \mathbf{Z}_k)\right) \\
	&\overset{(a)}{=} \mathsf{det}\left(\sigma_{\text{c}}^2\mathbf{I}_M+(\mathbf{U} + \mathbf{Z}_k)(\mathbf{U} + \mathbf{Z}_k)^{\top}\right)\\
	& = \mathsf{det}\left((\sigma_{\text{c}}^2+1)\mathbf{I}_M+\mathbf{Z}_k\mathbf{U}^{\top}+ \mathbf{U}\mathbf{Z}_k^{\top} + \mathbf{Z}_k\mathbf{Z}_k^{\top}\right),
\end{align*}
where (a) follows from Sylvester's determinant identity. The elements of $\mathbf{Z}_k$ are i.i.d. Gaussian $\mathcal{N}(0,\sigma_{\text{d}}^2)$, which means that $\mathbf{Z}_k$ has an isotropic matrix Gaussian distribution $\mathcal{MN}(\mathbf{0},\sigma_{\text{d}}^2\mathbf{I}_M,\mathbf{I}_N)$. There is a property that given $\mathbf{X}\sim\mathcal{MN}(\mathbf{0},\mathbf{\Sigma}_1,\mathbf{\Sigma}_2)$ and $\mathbf{Y}= \mathbf{M} + \mathbf{A}\mathbf{X}\mathbf{B}$, then $\mathbf{Y}\sim\mathcal{MN}(\mathbf{M},\mathbf{A}\mathbf{\Sigma}_1\mathbf{A}^{\top},\mathbf{B}^{\top}\mathbf{\Sigma}_2\mathbf{B})$. Therefore, both $\mathbf{Z}_k\mathbf{U}^{\top}$ and $\mathbf{U}\mathbf{Z}_k^{\top}$ are matrix Gaussian distributions $\mathcal{MN}(\mathbf{0}_M,\sigma_{\text{d}}^2\mathbf{I}_M,\mathbf{I}_M)$, which completes the proof. 

\subsection{Proof of Lemma~\ref{Lemma:ul_bounds}}\label{proof:ul_bounds}
First define $\mathbf{Q} = [\mathbf{U}^{\top}\;{\mathbf{U}^{\bot}}^{\top}]^{\top}$, where $\mathbf{U}^{\bot}$ represents the orthogonal complement of $\mathbf{U}$. Clearly, $\mathbf{Q}\in \mathcal{O}_N$ is an orthonormal matrix and thus $\mathbf{Q}^{-1}=\mathbf{Q}^{\top}$ always holds. According to the Woodbury matrix identity, the inverse matrix $(\mathbf{\Sigma}_k^{\prime})^{-1}$ is given by
\begin{align*}
	(\mathbf{\Sigma}_k^{\prime})^{-1} &=\left[\sigma_{\text{c}}^2\mathbf{I}_N+(\mathbf{U} + \mathbf{Z}_k)^{\top}(\mathbf{U} + \mathbf{Z}_k)\right]^{-1},\\
	& = \sigma_{\text{c}}^{-2}\mathbf{I}_N - \sigma_{\text{c}}^{-4}(\mathbf{U} + \mathbf{Z}_k)^{\top}\\
	&\qquad\cdot\left[\mathbf{I}_M+\sigma_{\text{c}}^{-2}(\mathbf{U} + \mathbf{Z}_k)(\mathbf{U} + \mathbf{Z}_k)^{\top}\right]^{-1}(\mathbf{U} + \mathbf{Z}_k),\\
	& = \sigma_{\text{c}}^{-2}\mathbf{I}_N - \sigma_{\text{c}}^{-4}\mathbf{Q}^{\top}\hat{\mathbf{Z}}_k^{\top}(\sigma_{\text{c}}^2\mathbf{I}_M +\hat{\mathbf{Z}}_k\hat{\mathbf{Z}}_k^{\top})^{-1}\hat{\mathbf{Z}}_k \mathbf{Q},
\end{align*}
where similarly $\hat{\mathbf{Z}}_k = [\mathbf{I}_M\;\mathbf{0}_{M,N-M}] + \mathbf{Z}_k\mathbf{Q}^{\top}$. Then, one can define the SVD of $\hat{\mathbf{Z}}_k$ as $\hat{\mathbf{Z}}_k = \hat{\mathbf{Q}}_k\hat{\mathbf{S}}_k\hat{\mathbf{P}}_k^{\top}$ where $\hat{\mathbf{Q}}_k$ is an $M$-dimensional orthonormal matrix, $\hat{\mathbf{P}}_k$ is an $N$-by-$M$ orthonormal matrix and $\hat{\mathbf{S}}_k = \mathsf{diag}\left(s_{1,k},s_{2,k},...,s_{M,k}\right)$. Therefore, the inverse of the covariance matrix is rewritten as
\begin{align*}
	(\mathbf{\Sigma}_k^{\prime})^{-1}
	&= \sigma_{\text{c}}^{-2}\mathbf{I}_N -\\ &\,\quad\mathbf{Q}^{\top}\hat{\mathbf{P}}_k \mathsf{diag}\left(\frac{\sigma_{\text{c}}^{-2}s_{1,k}^2}{\sigma_{\text{c}}^2 + s_{1,k}^2},...,\frac{\sigma_{\text{c}}^{-2}s_{M,k}^2}{\sigma_{\text{c}}^2 + s_{M,k}^2}\right)\hat{\mathbf{P}}_k^{\top} \mathbf{Q}, \\
	& 
	=  \sigma_{\text{c}}^{-2}\mathbf{I}_N - \sigma_{\text{c}}^{-2}\mathbf{Q}^{\top}\hat{\mathbf{P}}_k\hat{\mathbf{S}}_k^2\left(\sigma_{\text{c}}^{-2}\mathbf{I}_M + \hat{\mathbf{S}}_k^2\right)^{-1} \hat{\mathbf{P}}_k^{\top}\mathbf{Q}, 
\end{align*}
which finishes this proof. 

\subsection{Proof of Lemma~\ref{Lemma:mean:projection}}\label{proof:mean:projection}
$\hat{\mathbf{P}}_k$ is the principal eigenspace of $\hat{\mathbf{Z}}_k$ that follows from \emph{symmetric innovation}~\cite{JF2019estimation2019}. That is, let $l\in \{1,2,...,N\}$ and $\mathbf{D}_l = \mathbf{I}_N - 2\mathbf{e}_l\mathbf{e}_l^{\top}$ and we have $\hat{\mathbf{Z}}_k^{\top}\hat{\mathbf{Z}}_k\overset{\text{d}}{=}\mathbf{D}_l\hat{\mathbf{Z}}_k^{\top}\hat{\mathbf{Z}}_k\mathbf{D}_l$, where $\overset{\text{d}}{=} $ represents that both sides have the same distribution, so do their eigenspaces. One can observe that $\mathbf{D}_l\hat{\mathbf{P}}_k$ is the principal eigenspace of $\mathbf{D}_l\hat{\mathbf{Z}}_k^{\top}\hat{\mathbf{Z}}_k\mathbf{D}_l$. Therefore, we have $\mathsf{E}\left[\hat{\mathbf{P}}_k\hat{\mathbf{P}}_k^{\top}\right] = \mathsf{E}\left[\mathbf{D}_l\hat{\mathbf{P}}_k\hat{\mathbf{P}}_k^{\top}\mathbf{D}_l\right] = \mathbf{D}_l\mathsf{E}\left[\hat{\mathbf{P}}_k\hat{\mathbf{P}}_k^{\top}\right]\mathbf{D}_l$. As this equation holds for $l\in \{1,2,...,N\}$, we can conclude that $\mathsf{E}\left[\hat{\mathbf{P}}_k^{\top}\hat{\mathbf{P}}_k\right]$ is diagonal. Furthermore, define $\mathbf{p}_{k,i}$ as the $\hat{\mathbf{P}}_k$'s $i$-th row and we have the $i$-th diagonal element $\mu_i = \mathsf{E}\left[\Vert \mathbf{p}_{k,i} \Vert\right] \geq 0$.

\subsection{Proof of Corollary~\ref{Pro:unbiasedEwoCSI}} \label{proof:unbiasedEwoCSI}
Let $\mathbf{U}_{\hat{\mathbf{Y}}}$ denote the PCs of received signal $\hat{\mathbf{Y}} = \hat{\mathbf{H}}(\mathbf{U} + \mathbf{Z}) + \hat{\mathbf{W}}$, where the subscript $k$ is omitted. We aiming at proving $d(\tilde{\mathbf{U}}^{\star},\mathbf{U}) = 0$, where $ \tilde{\mathbf{U}}^{\star} = \mathcal{S}_M\left( \mathsf{E}[\mathbf{U}_{\hat{\mathbf{Y}}}^{\top}\mathbf{U}_{\hat{\mathbf{Y}}}]\right)$. 

First define 
\begin{align}
	\mathbf{M} 
	= &[\mathbf{Q}\hat{\mathbf{W}}^{\top}+([\mathbf{I}_M\;\mathbf{0}_{M,N-M}] + \mathbf{Z}\mathbf{Q}^{\top})^{\top}\hat{\mathbf{H}}^{\top}]\nonumber\\
	&\cdot[\hat{\mathbf{H}}([\mathbf{I}_M \mathbf{0}_{M,N-M}] + \mathbf{Z}\mathbf{Q}^{\top}) + \hat{\mathbf{W}}\mathbf{Q}^{\top}]\nonumber,
\end{align}
where $\mathbf{Q} = [\mathbf{U}^{\top}\;{\mathbf{U}^{\bot}}^{\top}]^{\top}$. Due to the isotropic property of matrix Gaussian $\mathbf{Z}$ and $\hat{\mathbf{W}}$, we have $\mathbf{Z}\mathbf{Q}$ and $\hat{\mathbf{W}}\mathbf{Q}$ still follow from matrix Gaussian endowed with \emph{symmetric innovation}. Similarly, let $l\in \{1,2,...,M\}$ and $\mathbf{D}_l = \mathbf{I} - 2\mathbf{e}_l\mathbf{e}_l^{\top}$ and we have $\mathbf{M}\overset{\text{d}}{=} \hat{\mathbf{M}}$, where $\hat{\mathbf{M}} = \mathbf{D}_l\mathbf{M}\mathbf{D}_l$. One can notice that $\mathbf{U}_{\hat{\mathbf{Y}}}$ is also the $M$-dimensional PCs of $\mathbf{Q}^{\top}\mathbf{M}\mathbf{Q}$. Then, let $\hat{\mathbf{U}}_{\hat{\mathbf{Y}}}$ denote the $M$-dimensional PCs of $\mathbf{Q}^{\top}\hat{\mathbf{M}}\mathbf{Q}$ and we have $\hat{\mathbf{U}}_{\hat{\mathbf{Y}}} = \mathbf{U}_{\hat{\mathbf{Y}}}\mathbf{Q}^{\top}\mathbf{D}_l\mathbf{Q}$ and $\mathbf{Q}\mathsf{E}[\mathbf{U}_{\hat{\mathbf{Y}}}^{\top}\mathbf{U}_{\hat{\mathbf{Y}}}]\mathbf{Q}^{\top} = \mathbf{Q}\mathsf{E}[\hat{\mathbf{U}}_{\hat{\mathbf{Y}}}^{\top}\hat{\mathbf{U}}_{\hat{\mathbf{Y}}}]\mathbf{Q}^{\top}= \mathbf{D}_l\mathbf{Q}\mathsf{E}[\mathbf{U}_{\hat{\mathbf{Y}}}^{\top}\mathbf{U}_{\hat{\mathbf{Y}}}]\mathbf{Q}^{\top}\mathbf{D}_l$.
As this equation holds for any $l$,  $\mathbf{Q}\mathsf{E}[\mathbf{U}_{\hat{\mathbf{Y}}}^{\top}\mathbf{U}_{\hat{\mathbf{Y}}}]\mathbf{Q}^{\top}$ is diagonal, meaning that the rows of $\mathbf{Q}$ comprise the eigenvectors of $\mathsf{E}[\mathbf{U}_{\hat{\mathbf{Y}}}^{\top}\mathbf{U}_{\hat{\mathbf{Y}}}]$. 

Then, to prove that the $M$-dimensional principal of $\mathsf{E}[\mathbf{U}_{\hat{\mathbf{Y}}}^{\top}\mathbf{U}_{\hat{\mathbf{Y}}}]$ matches with the ground-truth $\mathbf{U}$, we define $\mathbf{V} = \mathbf{U}_{\hat{\mathbf{Y}}}\mathbf{Q}^{\top} = [v_{i,j}]_{i=1:M,j=1:N}$ that can be regarded as $\mathbf{M}$'s PCs, where we have $\sum_{j=1}^N v_{i,j}^2 = 1, \forall i$. Using the definition yields $\mathbf{Q}\mathsf{E}[\mathbf{U}_{\hat{\mathbf{Y}}}^{\top}\mathbf{U}_{\hat{\mathbf{Y}}}]\mathbf{Q}^{\top} = \sum_{i=1}^{M}\mathsf{E}\left[\mathsf{diag}\left( v_{i,1}^2,...,v_{i,N}^2\right)\right]$. Then, due to the isotropic property of $\hat{\mathbf{H}}$, $\mathbf{Z}$ and $\hat{\mathbf{W}}$, the eigenvalue corresponding to the eigenvector $\mathbf{v}_i = [v_{i,1},v_{i,2},...,v_{i,N}]$ is given by
\begin{align}
	\lambda_i 
	& = \left\Vert [\hat{\mathbf{H}}([\mathbf{I}_M\;\mathbf{0}_{M,N-M}] + \mathbf{Z}\mathbf{Q}^{\top}) + \hat{\mathbf{W}}\mathbf{Q}^{\top}]\mathbf{v}_i^{\top}\right\Vert_2, \nonumber\\
	& \overset{\text{d}}{=} \left\Vert [\hat{\mathbf{H}}([\mathsf{diag}(|v_{i,1}|,...,|v_{i,M}|)\;\mathbf{0}_{M,N-M}] + \mathbf{z}) + \mathbf{w}]\right\Vert_2 \nonumber,
\end{align}
where Gaussian vectors $\mathbf{z}$ and $\mathbf{w}$ are independent of $\mathbf{v}_i$. It is clear that with larger $|v_{i,j}|, j\in \{1,...,M\}$, the above norm has a higher probability to get a larger value. Hence, $\mathsf{E}[\sum_{i=1}^M v_{i,j_1}^2] \geq \mathsf{E}[\sum_{i=1}^M v_{i,j_2}^2]$ holds for $j_1\in \{1,2,...,M\}$ and $j_2 \in \{M+1,M+2,...,N\}$. That is, the first $M$ diagonal elements of $\mathbf{Q}\mathsf{E}\left[\mathbf{U}_{\hat{\mathbf{Y}}}^{\top}\mathbf{U}_{\hat{\mathbf{Y}}}\right]\mathbf{Q}^{\top}$ are larger than the remaining diagonal elements and therefore the ground-truth $\mathbf{U}$ that aggregates the first $M$ rows of $\mathbf{Q}$ is the PCs of $\mathsf{E}\left[\mathbf{U}_{\hat{\mathbf{Y}}}^{\top}\mathbf{U}_{\hat{\mathbf{Y}}}\right]$.

\subsection{Proof of Lemma~\ref{Lemma:EwoCSIdistance}}\label{proof:EwoCSIdistance}
Let $\lambda^{\prime}$ and $\lambda^{\prime\prime}$ denote the eigenvalues of $\frac{1}{K}\sum_{k=1}^K \mathbf{U}_{\hat{\mathbf{Y}}_k}^{\top}\mathbf{U}_{\hat{\mathbf{Y}}_k}$ and $\mathbf{U}^{\top}\mathbf{U}$, respectively. Indeed, $\lambda_M^{\prime\prime} = 1$ and $\lambda_{M+1}^{\prime\prime} = 0$ hold. Based on the variant of the Davis-Kahan theorem (see~\cite[Corollary 3.1]{NIPS2013_81e5f81d}), we have
\begin{align*}
	d(\tilde{\mathbf{U}}^{\star},\mathbf{U}) 
	& \leq \frac{2\left\Vert \frac{1}{K}\sum_{k=1}^K \mathbf{U}_{\hat{\mathbf{Y}}_k}^{\top}\mathbf{U}_{\hat{\mathbf{Y}}_k} - \mathbf{U}^{\top}\mathbf{U} \right\Vert_F}{\max(\lambda_M^{\prime}-\lambda_{M+1}^{\prime},\lambda_M^{\prime\prime}-\lambda_{M+1}^{\prime\prime})},\\
	& \leq 2\left\Vert \frac{1}{K}\sum_{k=1}^K \mathbf{U}_{\hat{\mathbf{Y}}_k}^{\top}\mathbf{U}_{\hat{\mathbf{Y}}_k} - \mathbf{U}^{\top}\mathbf{U} \right\Vert_F,
\end{align*}

As a result, the square subspace distance is then upper bounded as $d^2(\tilde{\mathbf{U}}^{\star},\mathbf{U}) \leq \frac{4}{K}\sum_{k=1}^Kd^2(\mathbf{U}_{\hat{\mathbf{Y}}_k},\mathbf{U}) - \frac{2}{K^2}\sum_{k,j=1,k\neq j}^Kd^2(\mathbf{U}_{\hat{\mathbf{Y}}_k},\mathbf{U}_{\hat{\mathbf{Y}}_j})$.
We further define $\mathbf{E}_k = \mathbf{Z}_k + \hat{\mathbf{H}}_k^+\hat{\mathbf{W}}_k$ and let $\epsilon  = C_2[\sigma_{\text{d}}^2 + \frac{1}{2N_{\text{r}}-M-1}\sigma_{\text{c}}^2]$ with a constant $C_2$ such that $\Vert\epsilon^{-1}\mathbf{E}_k\Vert_2\leq 1$ is almost sure. Then, based on the Lemma~\ref{Lemma:subspaceDistanceBetweenTwoPerturbedSubspace} and \eqref{eq:approximation}, we have $d^2(\mathbf{U}_{\hat{\mathbf{Y}}_k},\mathbf{U}) 
= 2\mathsf{Tr}\left(\mathbf{\Delta}_{1,k}\mathbf{U^{\bot}}^{\top}\mathbf{U}^{\bot}\mathbf{\Delta}_{1,k}^{\top}\right) + O\left(\epsilon^3\right)$ and $d^2(\mathbf{U}_{\hat{\mathbf{Y}}_k},\mathbf{U}_{\hat{\mathbf{Y}}_j}) = 2\mathsf{Tr}\left( \mathbf{\Delta}_{2,k,j}\mathbf{U^{\bot}}^{\top}\mathbf{U}^{\bot}\mathbf{\Delta}_{2,k,j}^{\top}\right) +  O\left(\epsilon^3\right)$, where $\mathbf{\Delta}_{1,k} = \mathbf{E}_k - \mathbf{O}$ and $\mathbf{\Delta}_{2,k,j} = \mathbf{E}_k - \mathbf{E}_j$. Conditioned on matrices $\{\hat{\mathbf{H}}_k\}_{k=1}^{K}$, $\mathbf{E}_k\sim \mathcal{MN}[\mathbf{0},\sigma_{\text{d}}^2\mathbf{I}_M+\sigma_{\text{c}}^2\hat{\mathbf{H}}_k^+(\hat{\mathbf{H}}_k^{+})^{\top},\mathbf{I}_N]$. Therefore, we have $\mathbf{\Delta}_{1,k}\sim \mathcal{MN}[\mathbf{0},\sigma_{\text{d}}^2\mathbf{I}_M+\sigma_{\text{c}}^2\hat{\mathbf{H}}_k^+(\hat{\mathbf{H}}_k^{+})^{\top},\mathbf{I}_N]$, $\mathbf{\Delta}_{2,k,j}\sim\mathcal{MN}[\mathbf{0},2\sigma_{\text{d}}^2\mathbf{I}_M+\sigma_{\text{c}}^2(\hat{\mathbf{H}}_k^+(\hat{\mathbf{H}}_k^{+})^{\top} + \hat{\mathbf{H}}_j^+({\hat{\mathbf{H}}_j^{+})^{\top}}),\mathbf{I}_N]$.

\subsection{Proof of Theorem~\ref{The:upperBoundMSSDEwoCSI}}\label{proof:upperBoundMSSDEwoCSI}
Based on Lemma~\ref{Lemma:EwoCSIdistance}, in high SNR regime, the MSSD $\mathsf{E}[d^2(\tilde{\mathbf{U}}^{\star},\mathbf{U})]$ is bounded as
\begin{align}
	& \mathsf{E}[d^2(\tilde{\mathbf{U}}^{\star},\mathbf{U})], \nonumber\\
	& \leq \mathsf{E}\left[\frac{4}{K}\sum_{k=1}^K\mathsf{Tr}\left(\mathbf{\Delta}_{1,k}\mathbf{U^{\bot}}{}^{\top}\mathbf{U}^{\bot}\mathbf{\Delta}_{1,k}^{\top}\right)\right] \nonumber\\
	&\,\quad- \mathsf{E}\left[\frac{2}{K^2}\sum_{k,j=1,k\neq j}^K\mathsf{Tr}\left( \mathbf{\Delta}_{2,k,j}\mathbf{U^{\bot}}^{\top}\mathbf{U}^{\bot}\mathbf{\Delta}_{2,k,j}^{\top}\right)\right], \nonumber\\
	& = \frac{4(N-M)}{K}\left[M\sigma_{\text{d}}^2 + \sigma_{\text{c}}^2\mathsf{Tr}\left(\mathsf{E}\left[\hat{\mathbf{H}}_k^+(\hat{\mathbf{H}}_k^{+})^{\top}\right] \right)\right], \nonumber\\
    & \leq \frac{4M(N-M)}{K}\left(\sigma_{\text{d}}^2 + \frac{1}{2N_{\text{r}}-M-1}\sigma_{\text{c}}^2\right) \nonumber,
\end{align}
where the last inequality follows from the result in Appendix~\ref{proof:upperBoundMSSDEwCSI}. This completes the proof.

\bibliography{Ref}
\bibliographystyle{IEEEtran}
\end{document}